\documentclass[journal]{IEEEtran}
\usepackage{amsmath,amsfonts}

\usepackage{empheq}

\usepackage{amsthm,verbatim,amssymb,mathtools}
\usepackage[subtle]{savetrees}
\usepackage{subfigure}
\usepackage{hyperref}
\usepackage{algorithm}
\usepackage[noend]{algpseudocode}
\usepackage{adjustbox}
\usepackage{svg}
\usepackage{pgfplots}
\usepackage{scalefnt}
 \pgfplotsset{compat=1.9}

\usepackage{cite}
\usepackage{url} 

\usepackage{enumitem}

\newtheorem{theorem}{Theorem}

\newtheorem{lemma}[theorem]{Lemma}
\newtheorem{remark}{Remark}
\newtheorem{prop}{Proposition}

\title{A Low-rank Projected Proximal Gradient Method for Spectral Compressed Sensing}
\author{
    Xi Yao and {Wei Dai} ~\IEEEmembership{Member,~IEEE} 
    \thanks{Xi Yao (e-mail:
     x.yao19@imperial.ac.uk) and Wei Dai (e-mail:
     wei.dai1@imperial.ac.uk) are with the Department of Electrical and Electronic Engineering, Imperial College London, London SW7 2BU, United Kindom. }
}

\begin{document}\sloppy
%
\maketitle

\begin{abstract}
This paper presents a new approach to the recovery of a spectrally sparse signal (SSS) from partially observed entries, focusing on challenges posed by large-scale data and heavy noise environments. The SSS reconstruction can be formulated as a non-convex low-rank Hankel recovery problem. Traditional formulations for SSS recovery often suffer from reconstruction inaccuracies due to unequally weighted norms and over-relaxation of the Hankel structure in noisy conditions. Moreover, a critical limitation of standard proximal gradient (PG) methods for solving the optimization problem is their slow convergence. We overcome this by introducing a more accurate formulation and a  Low-rank Projected Proximal Gradient (LPPG) method, designed to efficiently converge to stationary points through a two-step process. The first step involves a modified PG approach, allowing for a constant step size independent of signal size, which significantly accelerates the gradient descent phase. The second step employs a subspace projection strategy, optimizing within a low-rank matrix space to further decrease the objective function. Both steps of the LPPG method are meticulously tailored to exploit the intrinsic low-rank and Hankel structures of the problem, thereby enhancing computational efficiency. Our numerical simulations reveal a substantial improvement in both the efficiency and recovery accuracy of the LPPG method compared to existing benchmark algorithms. This performance gain is particularly pronounced in scenarios with significant noise, demonstrating the method's robustness and applicability to large-scale SSS recovery tasks.
\end{abstract}
\section{Introduction}

A large class of practical applications features high-dimensional signals that can be modeled or approximated by a superposition of sparse spikes in the spectral domain and involve estimation of the signal from its time-domain samples. This spectrally sparse signal (SSS) arises in various applications, including medical imaging\cite{lustig2007sparse}, fluorescence microsopy\cite{schermelleh2010guide}, radar imaging\cite{potter2010sparsity}, channel estimation in wireless communications\cite{chi2013compressive}, and time series prediction\cite{gillard2022hankel},  among others. In many practical situations, only part of the signals can be measured because of hardware and physical constraints.  Hence, it is essential to explore spectral compressed sensing, which aims to recover SSSs from incomplete time-domain samples. It is worth noting that spectral compressed sensing is closely related to the problem of harmonic retrieval, which seeks to extract the underlying frequencies of a signal from a collection of its time-domain samples.  Conventionally, the frequencies of SSS can be extracted using methods such as MUSIC \cite{schmidt1986multiple}, Prony's method \cite{de1795essai} or a matrix pencil approach \cite{hua1992estimating}, exploiting the shift-invariance of the harmonic structure.

 If there is no damping factor in the signal, meaning that the signal does not weaken as time increases, spectral compressed sensing can be transformed into a conventional compressed sensing problem over a discrete domain \cite{donoho2006compressed,candes2006robust}. Various recovery algorithms have been developed using a finite discrete dictionary, including iterative thresholding \cite{blumensath2009iterative}, subspace pursuit \cite{dai2009subspace}, and CoSaMP \cite{needell2009cosamp}. However, the frequency leakage effect is unavoidable when the ground-truth frequencies are off the grid \cite{chi2011sensitivity,herman2010general}. The gridless approach avoids the issues of frequency discretization. The total-variation norm \cite{candes2014towards} and atomic norm \cite{tang2013compressed} minimization algorithms super-resolve the frequencies in a continuous domain. In these methods, the low-rank matrix arising in SSS recovery has a Toeplitz sub-matrix, and the exact recovery guarantee is based on the separation conditions.
 
\subsection{Related Gridless Work }

\textit{Convex approach:} Unfortunately, the aforementioned gridless methods cannot be applied to large-scale compressed sensing because of the high computational complexity of solving the equivalent semi-definite programming (SDP). Using the Vandermonde decomposition, Enhanced Matrix Completion (EMaC) \cite{chen2013spectral}, which originates from the traditional spectral estimation technique named Matrix Enhancement Matrix Pencil, proposed a low-rank Hankel matrix completion problem based on the shift-invariance property and spectral sparsity. However, it is well known that the structured low-rank approximation problem is an NP-hard global optimization problem\cite{gillis2011low}. For the Hankel low-rank approximation problem, the number of stationary points increases polynomially with size $n$ and exponentially with rank $r$\cite{ottaviani2014exact}. Inspired by low-rank completion \cite{candes2012exact}, EMaC solves the low-rank Hankel matrix completion problem via nuclear norm minimization as a convex problem. Later, Hankel matrix-based recovery approaches became more popular \cite{fazel2013hankel,xu2018sep} for large-scale problems due to their computational efficiency, as the involved matrix as a whole was Hankel structured and supports efficient decomposition. In an alternative approach, the Burer-Monteiro heuristic \cite{burer2005local,zhang2021spectrally} is applied where low-rank matrices are presented by a bilinear outer product to avoid explicit matrix singular value decomposition (SVD). However, it is noteworthy that the worst-case analysis in \cite{xu2018sep} shows that convex optimization can fail even when a single element is missing.

\textit{Non-convex approach: } The nuclear norm-based approach usually requires the full SVD of the enhanced matrix, which is unaffordable for large-scale problems. Therefore, many non-convex methods \cite{cadzow1988signal, ying2017hankel, cai2018spectral, cai2019fast, wang2021fast, gillard2022hankel} have been developed to solve low-rank Hankel matrix completion problems directly.  Cadzow's algorithm \cite{cadzow1988signal, gillard2022hankel} performs alternating projections between the set of Hankel matrices and the set of low-rank matrices. Although \cite{schost2016quadratically} proposed an extra Newton-type step to achieve quadratic convergence based on the common regularity intersection conditions \cite{lewis2008alternating, andersson2014new} of Cadzow's algorithm, its high computational cost prevents its practical application. Another acceleration method is inspired by Riemannian optimization \cite{wei2016guarantees} by adding a projection to the direct sum of the column and row spaces \cite{cai2019fast, wang2021fast}. Alternatively, a low-rank matrix/tensor can be written as an outer-product of low-rank matrices, and optimization can be performed directly on the component matrices to avoid explicit matrix decomposition \cite{cai2018spectral, ying2017hankel}. Among these, Fast Iterative Hard Thresholding (FIHT) \cite{cai2019fast} and Projected Gradient Descent (PGD) \cite{cai2018spectral} are the two most efficient algorithms with a convergence guarantee based on the Hankel incoherence property \cite{chen2013spectral}.

The enhanced low-rank matrix approximation in the context of spectrally sparse signal recovery faces two primary challenges: 1) the presence of a biased weighted norm due to repeated elements in the induced Hankel matrix, and 2) substantial computational costs associated with solving the structured low-rank approximation. The key issues stem from the condition number of the Hessian matrix of the quadratic terms in the objective function, which is proportional to the Hankel enforcement parameter $\beta$ and the signal size. This results in excessively slow convergence rates, especially when $\beta$ and the signal size are large. Additionally, the conventional approach of evaluating the rank function through SVD becomes impractical for large-scale problems due to its $O(n^3)$ complexity, leading to an unfeasible number of iterations for gradient-based methods even in modestly sized problems.

\subsection{Our Contributions}
To address these challenges, this work introduces a novel projected proximal gradient (LPPG) method, consisting of two meticulously crafted steps, to solve a purposefully designed objective function with equal weighted norms. This method demonstrates efficient convergence and enhanced performance in challenging scenarios, such as heavy noise and low sampling.
Our contributions are summarized as follows:
\begin{enumerate}
\item \textbf{A More Accurate Formulation:} The improvement of reconstruction recovery is supported by our formulation, which more accurately aligns with the original problem model in two key aspects. Unlike conventional Hankel-based algorithms, we represent the SSS in vector form rather than an augmented matrix form, thereby mitigating bias resulting from the different weights in the matrix $l_2$ form, as induced by the Hankel operator \cite{wang2021fast,gillard2022hankel}. This approach significantly improves reconstruction accuracy, proportionate to the signal size. Additionally, we introduce a Hankel enforcement parameter in the unconstrained optimization framework, allowing for adjustment based on noise levels and thereby enhancing reconstruction accuracy in heavily noisy environments.

\item \textbf{A Convergent Low-Rank Projected Proximal Gradient Algorithm:} While the proximal gradient (PG)-based approach ensures a monotonically decreasing function with a fixed step-size, the standard PG \cite{parikh2014proximal} is hindered by slow convergence, as its step-size is inversely proportional to the signal size. To accelerate convergence, we incorporate a modified PG step with an optimized, signal size-independent step-size. Furthermore, a novel low-rank matrix subspace projection step is employed to further reduce the number of iterations by decreasing the objective function, showcasing impressive recovery capabilities for low-sampling ratios and moderate orders of SSS. We quantify the reduction achieved by both steps, thereby ensuring guaranteed convergence.

\item \textbf{Efficient Implementations:} 
Beyond reducing the number of iterations in this two-step iterative model, we have also reduced the computational complexity of each iteration through two strategies. Firstly, the Hankel and low-rank structures are fully leveraged to decrease computational demands and storage requirements for large-sized variables in the iterative process. Secondly, optimization techniques are employed to transform the two-variable optimization problem into a single-variable paradigm using the subgradient condition. Consequently, we have reduced the computational complexity from $O(n^3)$ to $O(r^4n + r^3n\log n)$ per iteration.
\end{enumerate}

In fact, the Hankel enforcement parameter inversely propotional to the step-size in the gradient descent process. This parameter sets a strict upper limit for the step-size, thereby directly linking noise level and convergence speed. Unlike other descent methods that require complex estimation or line-search for step-size based on initial conditions, the LPPG method can efficiently determine it once the noise level is assessed. For instance, in noiseless situations, it's advisable to set the parameter for structural constraints near zero, as suggested in \cite{davenport2016overview}. This implies that a very large step-size can be selected for the quickest convergence when the Hankel enforcement parameter is small. Moreover, the LPPG method opens up promising prospects for the application of fast PG variations \cite{li2015accelerated,lee2014proximal} in spectral compressed sensing.

To validate our proposed LPPG method, we have conducted extensive numerical simulations. These simulations demonstrate the superiority of LPPG over existing benchmarks in four key areas: 1) it requires fewer samples to estimate the SSS with higher model order; 2) it ensures rapid convergence in noiseless scenarios; 3) it maintains robustness in the presence of heavy noise; 4) it improves accuracy by reducing the biased weighted norms for SSS.

This paper is an extension of the work \cite{yao2023projected} with efficient implementations, theoretical proofs, and broad simulations. 

\subsection{Notation}
Vectors and matrices are denoted by bold lowercase and uppercase letters, respectively, as \(\boldsymbol{x}\) and \(\boldsymbol{X}\). The transpose and the Hermitian of \(\boldsymbol{X}\) are represented by \(\boldsymbol{X}^\mathsf{T}\) and \(\boldsymbol{X}^\mathsf{H}\), respectively. The symbols \(\|\boldsymbol{x}\|_2\) and \(\|\boldsymbol{X}\|_F\) respectively denote the \(l_2\) norm and the Frobenius norm. The diagonal matrix formed from vector \(\boldsymbol{x}\) is denoted as \(\text{diag}(\boldsymbol{x})\). Notably, a Hankel matrix is a matrix with identical elements on ascending skew diagonals.

Operators are denoted by calligraphic letters. Specifically, \(\mathcal{I}\) denotes the identity operator and \(\mathcal{H}\) denotes the linear operator, mapping a vector \(\boldsymbol{x} \in \mathbb{C}^n\) to a Hankel matrix in \(\mathbb{C}^{p \times q}\) with \(p + q = n+1\) (where \(p \geq q\)). The entry \(\mathcal{H}\boldsymbol{x}[i,j]\) is given by \(\boldsymbol{x}_{i+j}\):  
\begin{equation}\label{eq:Hankel-mapping}
    \mathcal{H}\boldsymbol{x} =
    \begin{bmatrix}
        x_0 & x_1 & \dots & x_{q-1} \\ 
        x_1 & x_2 & \dots & x_{q} \\ 
        \vdots & & \vdots \\
        x_{p-1} & x_{p} & \dots & x_{n-1}
    \end{bmatrix}
    \in \mathbb{C}^{p \times q}.   
\end{equation}

Denote \(\mathcal{H}^*\) as the adjoint operator of \(\mathcal{H}\), which is a linear operator from \(p \times q\) matrices to \(n\)-dimensional vectors. For any matrix \(\boldsymbol{X} \in \mathbb{C}^{p \times q}\), it follows that \(\mathcal{H}^*\boldsymbol{X} = \big\{ \sum_{i+j=a} \boldsymbol{X}[i,j] \big\}_{a=0}^{n-1}\). Then, \(\mathcal{W} \coloneqq \mathcal{H}^*\mathcal{H}\) is a diagonal matrix equivalent to \(\text{diag}(\boldsymbol{w})\), where \(\boldsymbol{w} = [1, \dots, n_2-1, \underbrace{q, \cdots, q}_{p-q+1}, q-1, \dots, 1]^{\mathsf{T}}\). Moreover, the left inverse operator of \(\mathcal{H}\), \(\mathcal{S}\), which retrieves \(\boldsymbol{x}\) from its Hankel representation \(\mathcal{H}\boldsymbol{x}\), is defined as \(\mathcal{S} = \mathcal{W}^{-1}\mathcal{H}^*\).

\subsection{Organization}
The remainder of this paper is organized as follows. In Section~\ref{section: formulation}, we introduce our optimization formulation for spectral compressed sensing, with a focus on one-dimensional frequency models. This section includes a careful derivation of our objective function through a more precise nonconvex relaxation approach. In Section~\ref{section: method}, we present our proposed LPPG method, designed to find the critical points of the formulation defined. Detailed explanations of the LPPG algorithm are provided, highlighting its two-step iterative process: first, a modified proximal gradient step with a large, signal size-independent step size to ensure convergence to critical points; second, an optimization within the low-rank subspace to minimize the objective function. Additionally, we discuss the convergence properties and subgradient-based stopping criteria of the method. Section~\ref{section: CCA} go through the computational complexity and efficient implementation of the LPPG method. We also explore the extension of our method to multi-dimensional frequency models in Section~\ref{section: High-dimen}. In Section~\ref{sect: numerical}, we showcase the numerical validations of our method under various settings, demonstrating its efficiency and accuracy in comparison with existing benchmark methods. For readers interested in the technical intricacies, proofs of our results are provided in Section~\ref{sect: proofs}. Finally, Section~\ref{sect: conclusion} concludes the paper, summarizing our key contributions and discussing potential avenues for future research.

\section{The Optimization Formulation} \label{section: formulation}

We consider an order-\(r\) SSS\footnote{For ease of presentation, we focus on one-dimensional SSS in this paper. The proposed method can be extended to multilevel cases as discussed in Section~\ref{section: High-dimen}.} \(\boldsymbol{x} \in \mathbb{C}^n\) (\(r \ll n\)) which is a superposition of complex sinusoids:
\begin{equation}
    \label{eq:SSS}
    \boldsymbol{x} = \sum_{k=1}^{r} b_k \boldsymbol{y}(f_k,\tau_k;~n),
\end{equation}
where \(b_k \in \mathbb{C}\) is the complex coefficient, \(\boldsymbol{y}(f_k, \tau_k;~n) \coloneqq [1, e^{i2\pi f_k-\tau_k},\dots,e^{(i2\pi f_{k}-\tau_k)(n-1)}]^{\mathsf{T}} \in \mathbb{C}^n\), and \(f_k \in [0, 1), \tau_k\) denote the normalized frequency and damping factor of the \(k\)th spectral component\cite{cai2018spectral} separately. Due to physical constraints, typically only a part of \(\boldsymbol{x}\) is observed. Let \(\Omega \subset \{0, \dots, n-1 \}\) denote the set of indices of observed entries with \(|\Omega| \coloneqq m < n\), and let \(\mathcal{P}_{\Omega}\) be the corresponding sampling operator which acquires only the entries indexed by \(\Omega\) and sets all other entries to zero. The sampling rate is thus \(Sp \coloneqq m/n < 1\).

The challenge is to recover an SSS \(\boldsymbol{x}\) from its partial observations \(\boldsymbol{s}\), i.e.,
\begin{equation}
    \text{find} \quad \boldsymbol{x} \quad \text{s.t.} \quad \mathcal{P}_{\Omega}\boldsymbol{x} = \boldsymbol{s}.
\end{equation}
The Hankel matrix \(\mathcal{H}\boldsymbol{x}\) admits a Vandermonde decomposition:
\begin{equation}
    \mathcal{H}\boldsymbol{x} = \sum_{k=1}^r b_k \boldsymbol{y}(f_k, \tau_k;~p) \boldsymbol{y}(f_k,\tau_k;~q)^{\mathsf{T}} \in \mathbb{C}^{p \times q}.   
\end{equation}
Typically, we choose \(p \approx q \approx n/2\) to make \(\mathcal{H}\boldsymbol{x}\) close to a square matrix\cite{chen2013spectral,cai2018spectral}. It is clear that \(\text{rank} (\mathcal{H}\boldsymbol{x}) \le r \ll q\)\cite{hua1992estimating}.

Thus, SSS recovery can be transformed into a low-rank Hankel matrix recovery problem\cite{chen2013spectral,zhang2021spectrally,cai2018spectral,cai2019fast,ying2017hankel}. We introduce an enhanced Hankel matrix \(\boldsymbol{H} = \mathcal{H}\boldsymbol{x}\) to capture the data discrepancy in vector form while maintaining the Hankel structure. The rank-constrained nonconvex problem is formulated as: 
\begin{equation}
    \min_{\boldsymbol{H},\boldsymbol{x}}~ 
      \delta(\text{rank}(\boldsymbol{H}) \le r)
     + \frac{1}{2} \| \boldsymbol{s} - \mathcal{P}_{\Omega}\boldsymbol{x} \|^2_2 
    \quad \text{s.t.}~ 
    \boldsymbol{H} = \mathcal{H}\boldsymbol{x}, \label{eq:SSS-constrained}
\end{equation}
where \(\delta(\cdot)\) is the indicator function of rank-constrained matrices, taking values zero if true and infinity otherwise.

To address the complexities of solving the constrained nonconvex optimization problem, a relaxation approach is employed:
\begin{equation} 
    \label{eq:relaxed-objective}
    \min_{\boldsymbol{H},\boldsymbol{x}}~
    \delta(\text{rank}(\boldsymbol{H}) \le r) + \frac{1}{2} \| \boldsymbol{s} - \mathcal{P}_{\Omega}\boldsymbol{x} \|^2
    + \frac{\beta}{2} \| \boldsymbol{H} - \mathcal{H}\boldsymbol{x} \|^2_F + \frac{\alpha}{2}\|\boldsymbol{x}\|^2,
\end{equation}
where \(\beta > 0\) is the Hankel enforcement parameter, necessary to maintain the Hankel structure; and \(\alpha\) is a small positive scalar for regularization. The problem \eqref{eq:relaxed-objective} can be solved by the PG algorithm\cite{lee2012proximal} with a guarantee of convergence. However, direct application of PG leads to slow convergence caused by the unequal weighted norm \cite{zvonarev2015iterative,gillard2015stochastic} and high per-iteration computational complexity of \(O(n^3)\), primarily due to the SVD required by the rank constraint.

\begin{remark}
    For multidimensional SSS, where \(\boldsymbol{x} = \sum_k b_k \boldsymbol{y}(f_{1,k},\tau_{1,k};~n_1) \circ \cdots \circ \boldsymbol{y}(f_{d,k},\tau_{d,k};~n_d) \in \mathbb{C}^{n_1 \times \cdots \times n_d}\) with \(\circ\) denoting the outer product, a multilevel Hankel matrix is needed. Our method can be easily extended to such cases. For clarity, we focus on one-dimensional SSS in Sections~\ref{section: formulation} and \ref{section: method} for technical development but discuss extensions to multi-dimensional SSS in Section~\ref{section: High-dimen} and present simulation results in Section~\ref{sect: numerical}.
\end{remark}

We summarize the advantages of our formulation below:
\begin{enumerate}
    \item \textbf{The Hankel enforcement parameter \(\beta\):} We explicitly introduce \(\beta\) to adjust for different noise levels. This flexibility offers at least two advantages:
    1) Fast convergence speed for low-level noise; a small \(\beta\) can be chosen\cite{davenport2016overview}, and the step-size of gradient-based methods is inversely proportional to \(\beta\) as shown in \cite{cai2018spectral,cai2019fast};
    2) Accurate reconstruction for high-level noise; with unreliable data fidelity, a larger \(\beta\) better matches the original constrained optimization.
    
    \item \textbf{Data fidelity \(\| \boldsymbol{s} - \mathcal{P}_{\Omega}\boldsymbol{x} \|^2\):} The chosen data fidelity term differs from \(\|(\mathcal{H})(\boldsymbol{s} - \mathcal{P}_{\Omega}(\boldsymbol{x})) \|_F^2\) \cite{chen2013spectral,cai2018spectral,cai2019fast,wang2021fast,gillard2022hankel}, where different weighting coefficients for individual entries might introduce bias in the solution. Our approach ensures that noise terms are treated equally, avoiding unequal contributions and potential biases.
\end{enumerate}
The benefits of our formulation are numerically demonstrated in Section~\ref{sect: numerical}.

\section{A Low-rank Projected Proximal Gradient Method} \label{section: method}

Our proposed formulation is a nonconvex, unconstrained optimization problem incorporating nonsmooth terms. Our primary goal is to identify the critical points \((\boldsymbol{x}^{\star}, \boldsymbol{H}^{\star})\) of the objective function \eqref{eq:relaxed-objective}, based on the observed samples.

\subsection{The Modified Proximal Gradient (MPG) Step} \label{subsect: modified PG}
The proximal gradient (PG) method \cite{parikh2014proximal} is known to converge in nonconvex and nonsmooth optimization scenarios, provided that the step size is judiciously chosen. This gradient-descent based approach consistently reduces the objective function and exhibits greater robustness compared to alternating projection methods with random initialization\cite{cai2018spectral}.

The standard procedure for the PG method is outlined in Algorithm \ref{PGal}. This method tackles unconstrained optimization problems of the form:
\begin{equation}
    \min_{\boldsymbol{x}}~ F(\boldsymbol{x}) \coloneqq \min_{\boldsymbol{x}}~ f(\boldsymbol{x}) + g(\boldsymbol{x}),
\end{equation}
where \(f(\cdot)\) is a Lipschitz differentiable function\footnote{The gradient \(\nabla f\) is Lipschitz continuous, with the Lipschitz constant denoted by \(L_{\nabla f}\).}, and \(g(\cdot)\) is a proximable function. The proximal operator is defined as:
\begin{align}\label{prox}
    \text{prox}_{\gamma g}(\boldsymbol{v}) = \arg~\min_{\boldsymbol{x}} \left( g(\boldsymbol{x}) + \frac{1}{2\gamma} \|\boldsymbol{x} - \boldsymbol{v}\|^2  \right),
\end{align}
which is easy to solve. It is important to note that the step size \(0 < \gamma \le 1/L_f\) is inversely proportional to the Lipschitz constant \(L_f\).
\begin{algorithm}
    \caption{The Standard PG Method \cite{parikh2014proximal}}
    \label{PGal}
    \begin{algorithmic}
        \State Set \(\boldsymbol{x}_0\) and \(0 < \gamma \leq  \frac{1}{L_{f}}\) (\(L_f\): the Lipschitz constant of \(\nabla f\)).
        \For {\(k=0,1,\dots\)}
               \State Update \(\boldsymbol{x}_{k+1} \leftarrow \text{prox}_{\gamma g} \left( \boldsymbol{x}_k - \gamma \nabla f(\boldsymbol{x}_k)\right)\).
        \EndFor
    \end{algorithmic}
\end{algorithm}

Applying the PG method directly to the formulated function \eqref{eq:relaxed-objective} leads to defining
\begin{equation}
   f(\boldsymbol{H}, \boldsymbol{x}) 
    = \frac{1}{2} \| \boldsymbol{s} - \mathcal{P}_{\Omega}\boldsymbol{x} \|^2 
    + \frac{\beta}{2} \| \boldsymbol{H} - \mathcal{H}\boldsymbol{x} \|^{2}_{F} + \frac{\alpha}{2}\|\boldsymbol{x}\|^2. 
\end{equation}
The resulting Lipschitz constant of \(\nabla f\) is \(1 + \beta \min(p,q) + \alpha\), where \(p\) and \(q\) are derived from the Hankel mapping \eqref{eq:Hankel-mapping}. Consequently, the step-size becomes inversely proportional to the signal size, leading to unacceptable slow convergence for signals of moderate size.

In our methodology, we utilize the necessary condition for critical points, which dictates that the partial derivative of \(\boldsymbol{x}\) should be zero. This indicates that minimizing \(\boldsymbol{x}\) is a primary step. Notably, the variables \(\boldsymbol{x}\) and \(\boldsymbol{H}\) are linearly coupled only in the term \(\| \boldsymbol{H} - \mathcal{H}\boldsymbol{x} \|^{2}_{F}\). Consequently, a substitution approach is employed to identify the critical points \(\boldsymbol{x}^{\star}\) in terms of \(\boldsymbol{H}\). This substitution is crucial for reformulating the objective function to involve only the single variable \(\boldsymbol{H}\):
\begin{align}
    F(\boldsymbol{H}) & \coloneqq f(\boldsymbol{H}) + g(\boldsymbol{H}), \label{eq:F-x} \\
    f(\boldsymbol{H}) & \coloneqq \min_{\boldsymbol{x}}~ \frac{1}{2} \| \boldsymbol{s} - \mathcal{P}_{\Omega}\boldsymbol{x}\|^2 + \frac{\beta}{2} \| \boldsymbol{H} - \mathcal{H}\boldsymbol{x} \|^{2}_{F}+\frac{\alpha}{2}\|\boldsymbol{x}\|^2, \label{eq:f-x} \\
    g(\boldsymbol{H}) & \coloneqq \delta( \text{rank}(\boldsymbol{H}) \le r ). \label{eq:g-x}
\end{align}

Solving \(\nabla_{\boldsymbol{x}}\) of the equation yields the optimal \(\boldsymbol{x}^{\star}\):
\begin{equation}
    \label{eq:xoptimal}
    \boldsymbol{x}^{\star} = \mathcal{L} ( \mathcal{P}_{\Omega}^*\boldsymbol{s} + \beta\mathcal{H}^*\boldsymbol{H} ), \quad \mathcal{L} = (\alpha \mathcal{I}+\mathcal{P}_{\Omega}^*\mathcal{P}_{\Omega} + \beta\mathcal{H}^*\mathcal{H})^{-1}.
\end{equation}
\begin{remark}
    The operator \(\mathcal{L} = (\alpha \mathcal{I}+\mathcal{P}_{\Omega}^*\mathcal{P}_{\Omega} + \beta\mathcal{H}^*\mathcal{H})^{-1}\) exists and can be represented by diagonal matrices with positive entries. This is because \(\mathcal{H}^*\mathcal{H}\) scales each entry of an \(n\)-dimensional vector by \(w_a\), the number of elements in the \(a\)th skew-diagonal of a \(p \times q\) matrix, and \(\mathcal{P}_{\Omega}^*\mathcal{P}_{\Omega} = \mathcal{P}_{\Omega}\), represented by diagonal matrices with 1s on the observed entries and 0s elsewhere.
\end{remark}

A key aspect of formulation \eqref{eq:F-x} is the upper bound of the Lipschitz constant of \(\nabla f(\boldsymbol{H})\) by \(\beta\), independent of dimensionality. This allows for a larger step-size in the corresponding proximal gradient method, enabling faster convergence in large-scale problems. Proposition \ref{prop:Lipschitz-Constant} calculates the Lipschitz constant of \(\nabla f(\boldsymbol{H})\):
\begin{prop}
    \label{prop:Lipschitz-Constant}
    Denote the Lipschitz constant of \(\nabla f(\boldsymbol{H})\) as \(L_{\nabla f}\). It follows that:
    \begin{equation} \label{eq:gradient}
        \nabla f(\boldsymbol{H}) = - \beta \mathcal{H}\mathcal{L}\mathcal{P}^*_{\Omega}\boldsymbol{s} - \beta^2 \mathcal{H} \mathcal{L} \mathcal{H}^* \boldsymbol{H} + \beta \boldsymbol{H},
    \end{equation}
    and \(L_{\nabla f} < \beta\).
\end{prop}

\begin{remark}
    The improved upper bound on the Lipschitz constant, \(L_f\), permits a larger step size, \(\gamma = 1/\beta\), in iterative PG algorithms. This optimized step-size is pivotal for ensuring convergence.
\end{remark}

Efficient computations are another hallmark of our PG step. It primarily involves evaluations of \(\nabla f(\boldsymbol{H})\) as in \eqref{eq:gradient}, and the proximal operator:
\begin{equation} 
    \text{prox}_{\gamma g} \left(\boldsymbol{H}_k - \gamma \nabla f(\boldsymbol{H}_k)\right) = \mathcal{T}_r \left( \boldsymbol{H}_k - \gamma \nabla f(\boldsymbol{H}_k) \right) \label{eq:PGHankel},  
\end{equation}
where \(\mathcal{T}_r(\cdot)\) is the truncated SVD operator\cite{antonello2018proximal}. Significantly, \(\mathcal{T}_r\) operates on a combination of a Hankel matrix and a low-rank matrix, efficiently managed via the Lanczos method for calculating SVD's extreme triplets through matrix-vector multiplications. Additionally, the Hankel matrix-vector multiplications can be expedited from \(O(n^2)\) to \(O(n\log n)\) using the fast Fourier transform (FFT) \cite{lu2015fast}.

\subsection{The Low-Rank Projection Step} \label{subsect: SP}

In Hankel low-rank matrix completion, the number of stationary points scales exponentially with the rank \(r\) and linearly with the size \(n\) \cite{gillard2022hankel}. A critical aspect here is selecting optimal critical points that lead to more accurate reconstructions. In practices such as Cadzow's algorithm, a straightforward and cost-effective scalar correction is typically employed to reduce the cost function while preserving the matrix rank \cite{zvonarev2015iterative}.

Drawing inspiration from this approach, we have introduced a subspace projection step. This step aims to derive a new low-rank matrix \(\boldsymbol{H}_{k,\frac{1}{2}}\) that more effectively matches the Hankel structure and the observed samples, based on the output from the \(k\)-th iteration of the PG step \(\boldsymbol{H}_k\).

Specifically, the rank-\(r\) matrix \(\boldsymbol{H}_k\) is preprocessed into a reduced SVD decomposition: \(\boldsymbol{H}_k = \boldsymbol{U}_k \boldsymbol{\Sigma}_k \boldsymbol{V}_k^{\mathsf{H}}\), with \(\boldsymbol{U}_k \in \mathbb{C}^{p \times r}\) and \(\boldsymbol{V}_k \in \mathbb{C}^{q \times r}\) comprising orthonormal columns. By keeping \(\boldsymbol{U}_k\) and \(\boldsymbol{V}_k\) fixed, the optimization problem as defined in \eqref{eq:relaxed-objective} is reformulated as:
\begin{align}
    \min_{\boldsymbol{H}_k, \boldsymbol{x}} &\delta(\text{rank}(\boldsymbol{H}_k) \le r) + \frac{1}{2} \| \boldsymbol{s} - \mathcal{P}_{\Omega}\boldsymbol{x} \|^2 
    + \frac{\beta}{2} \| \boldsymbol{H}_k - \mathcal{H}\boldsymbol{x} \|^2_F + \frac{\alpha}{2}\|\boldsymbol{x}\|^2 \nonumber \\
    &= \min_{\boldsymbol{\Sigma},\boldsymbol{x}}~ \frac{1}{2} \| \boldsymbol{s} - \mathcal{P}_{\Omega}\boldsymbol{x} \|^2 
    + \frac{\beta}{2} \| \boldsymbol{U}_k\boldsymbol{\Sigma} \boldsymbol{V}_k^{\mathsf{H}} - \mathcal{H}\boldsymbol{x} \|^2_F + \frac{\alpha}{2}\|\boldsymbol{x}\|^2,
    \label{eq:projection-optimization-problem}
\end{align}
which constitutes a least-squares problem involving variables \(\boldsymbol{\Sigma}\)\footnote{Note that $\boldsymbol{\Sigma}$ is not assumed to be diagonal} and \(\boldsymbol{x}\), with a unique closed-form solution. However, the computational complexity of solving \eqref{eq:projection-optimization-problem}, generally \(O((n+r^2)^3)\). Our preliminary simulations show that this complexity can be too much for large $n$.

As we have done previously, we can reformulate the optimization problem \eqref{eq:projection-optimization-problem} into a single-variable formulation \eqref{eq:optimization-Sigma}, significantly reducing the computational complexity from \(O((n+r^2)^3)\) to \(O(r^3 n \log n + r^4 n)\). Specifically, we recast \eqref{eq:projection-optimization-problem} as follows:
\begin{align} \label{eq:optimization-Sigma}
    &\min_{\boldsymbol{\Sigma}_k}~L(\boldsymbol{\Sigma}_k), \text{ with } \nonumber \\ 
    &L(\boldsymbol{\Sigma}_k) = \min_{\boldsymbol{x}}
    \frac{1}{2} \| \boldsymbol{s} - \mathcal{P}_{\Omega}\boldsymbol{x} \|^2 
    + \frac{\beta}{2} \| \boldsymbol{U}_k \boldsymbol{\Sigma}_k \boldsymbol{V}_k^{\mathsf{H}} - \mathcal{H}\boldsymbol{x} \|^2_F + \frac{\alpha}{2}\|\boldsymbol{x}\|^2.
\end{align}
In this formulation, the objective function \(L(\boldsymbol{\Sigma}_k)\) depends solely on the variable \(\boldsymbol{\Sigma}_k\), yet its evaluation also involves optimizing \(\boldsymbol{x}\). The optimal solutions to \eqref{eq:optimization-Sigma} are delineated below.

\begin{prop} \label{pro: gradient_h}
    Let \(\mathcal{P}_{\boldsymbol{U}_k,\boldsymbol{V}_k}(\boldsymbol{\Sigma}) = \boldsymbol{U}_k \boldsymbol{\Sigma} \boldsymbol{V}_k^{\mathsf{H}}\), and \(\mathcal{P}^*_{\boldsymbol{U}_k,\boldsymbol{V}_k}(\boldsymbol{H}) = \boldsymbol{U}_k^{\mathsf{H}} \boldsymbol{H} \boldsymbol{V}_k\). The optimal solutions to \eqref{eq:optimization-Sigma} and \eqref{eq:projection-optimization-problem} are:
    \begin{align}
        \boldsymbol{x}^{\dagger}
        &= (\alpha \mathcal{I} + \mathcal{P}_{\Omega}^*\mathcal{P}_{\Omega} + \beta\mathcal{H}^*\mathcal{H})^{-1} \left( \mathcal{P}_{\Omega}^* \boldsymbol{s} + \beta\mathcal{H}^* \mathcal{P}_{\boldsymbol{U}_k,\boldsymbol{V}_k} \boldsymbol{\Sigma}^\star_k \right), \label{eq:x-proj-closedform} \\
        \boldsymbol{\Sigma}_k^{\star}
        &= \left( \mathcal{I} - \beta \mathcal{P}^*_{\boldsymbol{U}_k,\boldsymbol{V}_k} \mathcal{H} \mathcal{L} \mathcal{H}^*\mathcal{P}_{\boldsymbol{U}_k,\boldsymbol{V}_k} \right)^{-1} \left( \mathcal{P}^*_{\boldsymbol{U}_k,\boldsymbol{V}_k} \mathcal{H} \mathcal{L}\mathcal{P}^*_{\Omega} \boldsymbol{s} \right). \label{eq:Sigma-proj-closedform}
    \end{align}  
\end{prop}

\begin{remark}
   Due to the regularization term \(\frac{\alpha}{2}\|\boldsymbol{x}\|^2\), the matrix \(\mathcal{I} - \beta \mathcal{P}^*_{\boldsymbol{U}_k,\boldsymbol{V}_k} \mathcal{H} \mathcal{L} \mathcal{H}^*\mathcal{P}_{\boldsymbol{U}_k,\boldsymbol{V}_k}\) is positive definite, ensuring the existence of its inverse linear mapping.
\end{remark}

\begin{prop} \label{Theo_SP_reduction}
    Denote \(\boldsymbol{H}_{k,\frac{1}{2}} = \boldsymbol{U}_k \boldsymbol{\Sigma}_k^{\star} \boldsymbol{V}_k^{\mathsf{H}}\). The low-rank projection step satisfies the sufficient decrease property:
    \begin{equation}
        F(\boldsymbol{H}_{k,\frac{1}{2}}) \leq F(\boldsymbol{H}_k) - \left(\frac{\alpha}{\alpha + \beta}\right) \| \boldsymbol{H}_{k,\frac{1}{2}} - \boldsymbol{H}_{k}\|^2_F.
    \end{equation}
\end{prop}

The computation of Equations \eqref{eq:x-proj-closedform} and \eqref{eq:Sigma-proj-closedform} is efficiently facilitated by the Hankel and low-rank structure. The reduction in variable size from \(n+r^2\) to \(r^2\) significantly reduces computational complexity, making it feasible to find the optimal values using the conjugate gradient (CG) method for linear mapping. The computational complexity for the low-rank projection step is \(O(r^3n\log n + r^4n)\), with a detailed analysis included in Section \ref{section: CCA}.

This low-rank projection step can be viewed as an adjustment of SVD to enhance data fidelity and maintain the Hankel structure following each MPG step.

\subsection{LPPG Algorithm and Convergence}\label{subsection: algorithm}

We integrate the modified proximal gradient (MPG) with an optimized step-size and the low-rank subspace projection into our LPPG algorithm, as detailed in Algorithm \ref{LPPG}.

\begin{algorithm}
\caption{LPPG for equal Weighted Spectrally Sparse Signal Reconstruction}
\label{LPPG}
\textbf{Input:} \(\boldsymbol{s}\), \(\beta\), \(\alpha\), \(\epsilon\), \(\gamma = \frac{1}{\beta}\), \(\boldsymbol{H}_0 \leftarrow \mathcal{T}_r\mathcal{H}\boldsymbol{s}\).
\begin{algorithmic}
    \While{\(\| \partial F(\boldsymbol{H}_{k}) \|_F \geq \epsilon \| \boldsymbol{H}_{k} \|_F\)}	
        \State Compute \(\boldsymbol{H}_{k,\frac{1}{2}} \leftarrow \boldsymbol{U}_k \boldsymbol{\Sigma^{\star}}_k \boldsymbol{V}^{\mathsf{H}}_k\) via subspace projection of \(\boldsymbol{H}_{k}\).
        \State Compute \(\boldsymbol{H}_{k+1} \leftarrow \mathcal{T}_r \left(\boldsymbol{H}_{k,\frac{1}{2}}-\gamma  \nabla f(\boldsymbol{H}_{k,\frac{1}{2}}) \right)\).
    \EndWhile
\end{algorithmic}
\textbf{Output:} \(\boldsymbol{H}_k\) from the last iteration, \(\boldsymbol{x}^{\dag} \leftarrow \mathcal{L} (\mathcal{P}^*_{\Omega}\boldsymbol{s} + \beta\mathcal{H}^* \boldsymbol{H}_{k})\).
\end{algorithm}

The objective function \eqref{eq:F-x} monotonously decreases after each iteration step due to the optimization of the subproblem. We also present the convergence result of the LPPG algorithm in Theorem \ref{theo: convergence}.

\begin{theorem} \label{theo: convergence}
The sequences \(\{\boldsymbol{H}_{k,\frac{1}{2}}\}\) and \(\{\boldsymbol{H}_k\}\) generated in Algorithm \ref{LPPG} are bounded. If \(\boldsymbol{H}^*\) is any accumulation point of \(\{\boldsymbol{H}_k\}\), then \(0 \in \partial F(\boldsymbol{H}^*)\), i.e., \(\boldsymbol{H}^*\) is a critical point. We have:
\begin{equation} \label{eq: partialFlimit}
    -\nabla f(\boldsymbol{H}_{k,\frac{1}{2}}) + \nabla f(\boldsymbol{H}_{k+1}) - \frac{1}{\gamma} (\boldsymbol{H}_{k+1}-\boldsymbol{H}_{k,\frac{1}{2}}) \in \partial F(\boldsymbol{H}_{k+1}),
\end{equation}
and
\begin{align} \label{eq: boundgradient}
    &\| -\nabla f(\boldsymbol{H}_{k,\frac{1}{2}}) + \nabla f(\boldsymbol{H}_{k+1}) - \frac{1}{\gamma} (\boldsymbol{H}_{k+1}-\boldsymbol{H}_{k,\frac{1}{2}})\|_F \\
    \leq &\left(\frac{1}{\gamma} + L_f\right) \| \boldsymbol{H}_{k+1}-\boldsymbol{H}_{k,\frac{1}{2}} \|_F \rightarrow 0 \quad  \text{ as } \quad k \rightarrow \infty.
\end{align}

The convergence rate is given by:
\[
\min_{i=0,\cdots,K} ~  \|\partial F(\boldsymbol{H}_{i+1})  \|_F^2 \leq \frac{c_0}{K+1} ( F(\boldsymbol{H}_{0}) - F^*),
\]
where \(c_0 = \frac{(\beta-L_f)\beta^4}{(\alpha+\beta)^2}\).
\end{theorem}

The proof involves three steps: 1) establishing sufficient decrease; 2) proving that \(\boldsymbol{H}^*\) is a critical point; and 3) determining the convergence rate. Detailed proof is provided in Section \ref{subsection：convergence}.

\begin{remark}
    The convergence primarily stems from the MPG step, showing that \(\|\partial F  \|_F \le \left(\frac{1}{\gamma} + L_{\nabla f}\right)\| \boldsymbol{H}_{k+1}- \boldsymbol{H}_{k,\frac{1}{2}}\|_F\), which approaches zero. The additional low-rank subspace projection further decreases the objective function, leading to faster convergence and better reconstruction empirically, as it aligns with the observed samples and Hankel structure.
\end{remark}

\subsection{Stopping Criteria} \label{subsect: stopping criteria}

Previous studies on non-convex approaches, such as \cite{cai2018spectral,cai2019fast,wang2021fast,cadzow1988signal,gillard2010cadzow}, typically employed stopping criteria based on the relative difference between adjacent iterations, calculated as \(\|\boldsymbol{H}_{k+1}-\boldsymbol{H}_{k}\|_2/\| \boldsymbol{H}_{k} \|_2\). This criterion is effective when the initial guess lies within the basin of attraction of the global minimum. However, in practice, particularly when frequencies to be estimated are not well-separated, this criterion may be inadequate due to the potential influence of small steps implicitly leading the process.

Our LPPG method addresses this issue by directly examining the convergence to critical points through the subgradient of \(F\). To ascertain the accumulation point of \(\{ \boldsymbol{H}_k \}\), we ensure that \(\partial F(\boldsymbol{H}_k)\) approaches zero. Given that \(g\) is not differentiable, we utilized the optimality condition of the proximal operator \eqref{eq:PGHankel} to deduce:
\begin{align}
    -\frac{1}{\gamma} (\boldsymbol{H}_{k+1}-(\boldsymbol{H}_{k,\frac{1}{2}} - \gamma \nabla f(\boldsymbol{H}_{k,\frac{1}{2}})  )) \in \partial g(\boldsymbol{H}_{k+1}) .
\end{align}
Consequently, we have :
\begin{align} 
    &-\frac{1}{\gamma} (\boldsymbol{H}_{k+1}-(\boldsymbol{H}_{k,\frac{1}{2}} - \gamma \nabla f(\boldsymbol{H}_{k,\frac{1}{2}})  )) + \nabla f(\boldsymbol{H}_{k+1}) \nonumber \\
    =& \frac{1}{\gamma}  (\boldsymbol{H}_{k,\frac{1}{2}}- \boldsymbol{H}_{k+1}  )  + ( \nabla f(\boldsymbol{H}_{k+1}) - \nabla f(\boldsymbol{H}_{k,\frac{1}{2}})) \nonumber \\
    = & ((\frac{1}{\gamma} - \beta )\mathcal{I}+\beta^2 \mathcal{H}\mathcal{L}\mathcal{H}^*)(\boldsymbol{H}_{k,\frac{1}{2}}- \boldsymbol{H}_{k+1}  )\nonumber \\ = & \beta^2 \mathcal{H}\mathcal{L}\mathcal{H}^*(\boldsymbol{H}_{k,\frac{1}{2}}- \boldsymbol{H}_{k+1}  )  \in \partial F(\boldsymbol{H}_{k+1}). \label{eq:SC_F}
\end{align}
This formulation allows efficient verification of \( \beta^2 \mathcal{H}\mathcal{L}\mathcal{H}^*(\boldsymbol{H}_{k,\frac{1}{2}}- \boldsymbol{H}_{k+1}  ) \) as detailed in Section \ref{section: CCA}.

The stopping criteria correlate with the choice of \(\beta\): in noiseless scenarios, we expect perfect recovery with large steps, leading the subgradient of \(F\) to approach zero. In contrast, under heavy noise, the subgradient of \(F\) may oscillate around local minima, prompting an earlier termination.

\section{Computational Complexity Analysis} \label{section: CCA}

This section elucidates the leading computational complexity of each iteration of our proposed LPPG algorithm as \(O(r^4n + r^3n\log n)\), significantly lower than the \(O(n^3)\) complexity typically associated with large size SSS signals. This reduction is achieved by exploiting the low-rank and Hankel structures inherent in the problem, rendering the LPPG algorithm effective for large-scale SSS recovery.

We base our computational complexity analysis on several fundamental operations relevant to our algorithm:

\begin{itemize}
    \item \textbf{Hankel Matrix-Vector Multiplication:} For a \(p \times q\) Hankel matrix mapping from an array of size \(n\), the matrix-vector multiplication can be implemented efficiently using FFT with \(O(n \log n)\) flops \cite{golub2013matrix}.
    
    \item \textbf{Low-Rank Matrix-Vector Multiplication:} The complexity of multiplying a rank-$r$ square matrix of size $n \times n$ with a vector is $O(rn)$. Consider a rank \(r\) matrix \(\boldsymbol{A} \in \mathbb{C}^{n \times n}\) represented in SVD form as \(\boldsymbol{U} \boldsymbol{\Sigma V}^{\mathsf{H}}\), where \(\boldsymbol{U}\) and \(\boldsymbol{V}^{\mathsf{H}}\) have \(r\) columns and rows, respectively. The complexity of multiplying \(\boldsymbol{A}\) by a vector \(\boldsymbol{x}\) is:
   \begin{equation}
        \underbrace{\boldsymbol{A}}_{n\times n}\times \underbrace{\boldsymbol{x}}_{n\times 1} = \underbrace{\boldsymbol{U}}_{n\times r} \times \large [\underbrace{\boldsymbol{\Sigma }}_{r \times r}\times  \large( \underbrace{\boldsymbol{V}^{\mathsf{H}}}_{r \times n} \times \underbrace{\boldsymbol{x}}_{n\times 1} \large) \large]
    \end{equation} 
    In the LPPG process, the SVD of the low-rank matrix is explicitly computed, leading to a complexity of \(O(rn)\) for \(\boldsymbol{V}^{\mathsf{H}}\boldsymbol{x}\), \(O(r^2)\) for \(\boldsymbol{\Sigma}\) times \(\boldsymbol{V}^{\mathsf{H}}\boldsymbol{x}\), and \(O(rn)\) for \(\boldsymbol{U}\) times the resulting product ($r \leq n$).

    \item \textbf{Hankel Mapping \(\mathcal{H}\):} The \(p \times q\) Hankel matrix \(\mathcal{H}x\) is efficiently stored in \(\boldsymbol{x}\) with \(O(n)\) complexity due to the repetition of elements in \(\mathcal{H}x\). Thus, addition, scaling, and inner products within the Hankel subspace require \(O(n)\) flops.

    \item \textbf{Adjoint of Hankel Mapping \(\mathcal{H}^*\):} The adjoint operator \(\mathcal{H}^*\) transforms \(p \times q\) matrices into vectors of length \(n\). It is noticed that during the iterations of the LPPG algorithm, $\{\boldsymbol{H}_k \}$ and $ \{\boldsymbol{H}_{k,\frac{1}{2} }\}$ are rank-$r$ matrices, admitting a decomposition $\boldsymbol{H} = \boldsymbol{U} \boldsymbol{\Sigma} \boldsymbol{V}^{\mathsf{H}}$. In our algorithm, $\mathcal{H}^*$ always operates after $\mathcal{P}_{\boldsymbol{U},\boldsymbol{V}}$ as:
    \begin{align}
      & \mathcal{H}^* \mathcal{P}_{\boldsymbol{U},\boldsymbol{V}} \boldsymbol{\Sigma} = \mathcal{H}^* \underbrace{\boldsymbol{U} \boldsymbol{\Sigma}}_{\hat{\boldsymbol{U}}} \boldsymbol{V}^{\mathsf{H}} \nonumber \\
      =& \mathcal{H}^*\hat{\boldsymbol{U}} \boldsymbol{I} \boldsymbol{V}^{\mathsf{H}}
      =\sum_{i=1}^r  \mathcal{H}^* [ \hat{\boldsymbol{U}}(:,i) \boldsymbol{V}(:,i)^{\mathsf{H}}].
    \end{align}
    Noting that 
    \begin{equation} \label{eq:convolution}
          \left[ \mathcal{H}^* [ \hat{\boldsymbol{U}}(:,i) \boldsymbol{V}(:,i)^{\mathsf{H}}] \right]_a = \sum_{p=0,q=0}^{p+q=a} \hat{\boldsymbol{U}}(p,i) \bar{\boldsymbol{V}}(q,i),
    \end{equation}
    where $a=0,1,\dots,n-1$. \eqref{eq:convolution} can be computed by fast convolution using $O(n\log n)$ flops and $\boldsymbol{U} \boldsymbol{\Sigma}$ needs $O(nr^2)$ flops. Therefore, 
    $\mathcal{H}^* \mathcal{P}_{\boldsymbol{U},\boldsymbol{V}} \boldsymbol{\Sigma}$ requires $O(nr^2 + rn \log n)$ flops. 
\end{itemize}

With the aforementioned preliminaries, it is evident that the LPPG algorithm effectively balances computational efficiency with its capability to handle large-scale SSS recovery problems, leveraging the unique structures of the matrices involved. Building on these efficient calculations, we analyze the computational complexity in each iteration, which encompasses the following three parts:

\begin{itemize}
    \item \textbf{Subspace Projection Step:} This step involves solving the following equation for \(\boldsymbol{\Sigma}\):
    \begin{equation}
      \left( \mathcal{I}- \beta \mathcal{P}^*_{\boldsymbol{U},\boldsymbol{V}} \mathcal{H} \mathcal{L} \mathcal{H^*}\mathcal{P}_{\boldsymbol{U},\boldsymbol{V}} \right) \boldsymbol{\Sigma} = \mathcal{P}^*_{\boldsymbol{U},\boldsymbol{V}} \mathcal{H} \mathcal{L}\mathcal{P}^*_{\Omega} \boldsymbol{s}. \label{eq:CG}
    \end{equation}
    This general linear equation can be solved using the conjugate gradient (CG) method \cite{hestenes1952methods} with at most \(r^2\) iterations, where \(r^2\) is the size of the variable \(vec(\boldsymbol{\Sigma})\). We have shown that each iteration of the linear mapping in \eqref{eq:CG} requires \(O(nr^2+rn\log n)\) floating-point operations (flops). Hence, the total computational complexity for the subspace projection step is \(O(nr^4+r^3n\log n)\).

    \item \textbf{Modified PG Step:} 
    As delineated in the proximal mapping \eqref{eq:PGHankel}, the modified PG step entails solving a rank-\(r\) truncated singular value decomposition (SVD) of a linear combination of a Hankel matrix and a low-rank matrix. Constructing this Hankel matrix demands \(O(r^2n + rn\log n)\) flops. We implement the fast \(r\)-truncated SVD for the linear combination matrix using the Lanczos method \cite{lanczos1950iteration,baglama2005augmented} with \(O(r n \log n + r^2 n)\) flops, based on the fast matrix-vector multiplication of the Hankel matrix and low-rank matrix. Therefore, the total computational complexity of the modified PG step is \(O(r^2 n + r n \log n)\).

    \item \textbf{Stopping Criteria:} As indicated in \eqref{eq:SC_F}, after projecting to Hankel space with \(O(r^2n+ r n\log n)\) flops via \(\mathcal{H}^*\), we can compute \(\| \partial F \|_F\) in \(O(n)\) flops. Thus, our stopping criteria can be efficiently verified with \(O(r^2n + r n\log n)\) flops.
\end{itemize}

In summary, the computational complexity of our LPPG algorithm is \(O(r^4n+r^3n\log n)\) per iteration. This relatively low complexity in terms of \(n\) renders our algorithm well-suited for large-scale SSS recovery problems.

\section{High Dimensional Problems} \label{section: High-dimen}

Building on the Hankel-based approaches of \cite{chen2013spectral,cai2018spectral,cai2019fast,wang2021fast}, our results can be extended to higher dimensions by leveraging the Hankel structure of multidimensional spectrally sparse signals. We illustrate this extension through a two-dimensional example, although the principles apply similarly in \(d\) dimensions. A two-dimensional spectrally sparse array \(\boldsymbol{x} \in \mathbb{C}^{n_1 \times n_2 }\) can be expressed as:
\begin{equation}
    \boldsymbol{x} = \sum_{k=1}^{r} b_k \boldsymbol{y}(f_{1k},\tau_{1k};~n_1) \circ \boldsymbol{y}(f_{2k},\tau_{2k};~n_2)  \in \mathbb{C}^{n_1 \times n_2 },
\end{equation}
where \(b_k \in \mathbb{C}\) are complex amplitudes,  \(\boldsymbol{f}_k=(f_{1k},f_{2k}) \in [0,1)^2\) denote the frequencies and \(\tau_{1k}, \tau_{2k} \geq 0\) are the damping factors respectively. Given a set of indices \(\boldsymbol{\Omega}=\{ (\omega_1, \omega_2) \in [n_1] \times [n_2] \}\) for the known entries of \(\boldsymbol{x}\), the challenge is to reconstruct \(\boldsymbol{x}\) from \(\mathcal{P}_{\Omega}(\boldsymbol{x})\), which can be approached by exploring the low-rank Hankel structures as in one dimension.

The corresponding Hankel matrix for \(\boldsymbol{x}\) is constructed recursively:
\[
\mathcal{H}(\boldsymbol{x})
    =
    \begin{bmatrix}
        \mathcal{H}(\boldsymbol{x}[:,0]) & \mathcal{H}(\boldsymbol{x}[:,1]) & \dots & \mathcal{H}(\boldsymbol{x}[:,q_2-1])\\ 
        \mathcal{H}(\boldsymbol{x}[:,1]) & \mathcal{H}(\boldsymbol{x}[:,2]) & \dots & \mathcal{H}(\boldsymbol{x}[:,q_2])\\ 
        \vdots & \vdots & \ddots & \vdots \\
        \mathcal{H}(\boldsymbol{x}[:,p_2-1]) & \mathcal{H}(\boldsymbol{x}[:,p_2]) & \dots & \mathcal{H}(\boldsymbol{x}[:,n_2-1])
    \end{bmatrix}
    \in \mathbb{C}^{p_1p_2 \times q_1q_2},
\]
where $\boldsymbol{x}[:j]$, $0 \leq j < n_2-1$ is the $j-$th slice of $\boldsymbol{x}$ and 
\[
\mathcal{H}(\boldsymbol{x}[:,j])
    =
    \begin{bmatrix}
        \mathcal{H}(\boldsymbol{x}[0,j]) & \dots & \mathcal{H}(\boldsymbol{x}[q_1-1,j])\\ 
        \vdots & \ddots & \vdots \\
        \mathcal{H}(\boldsymbol{x}[p_1-1,j]) & \dots & \mathcal{H}(\boldsymbol{x}[n_1-1,j])
    \end{bmatrix},
\]
with \(p_1+q_1=n_1+1\), \(p_2+q_2=n_2+1\), and \(\mathcal{H}\boldsymbol{x}[:,j]    \in \mathbb{C}^{p_1 \times q_1}\).

An explicit formula for \(\mathcal{H}\boldsymbol{x}\) is given by:
\[
[\mathcal{H}\boldsymbol{x}]_{uv} = \boldsymbol{x}[i,j],
\]
where 
\begin{alignat*}{2}
    u&=u_1+u_2*p_1              , \quad &&v = v_1+v_2 \cdot (n_1-p_1+1) \\
    i&=u_1+v_1 , \quad                  &&j = u_2 + v_2.
\end{alignat*}

Moreover, \(\mathcal{H}\boldsymbol{x}\) admits a Vandermonde decomposition \(\mathcal{H}\boldsymbol{x}= \boldsymbol{E_L}\boldsymbol{D}\boldsymbol{E_R}^{\mathsf{T}}\), with the \(k\)-th columns of \(\boldsymbol{E_L}\) and \(\boldsymbol{E_R}\) given by:
\begin{align*}
    \boldsymbol{E_L}[:,k] & = \{  \boldsymbol{y}(f_{1k};~n_1)[i] \cdot \boldsymbol{z}(f_{2k};~n_2)[j] \; | \; 0\leq i\leq p_1-1 , 0\leq j\leq p_2-1  \} \\
    \boldsymbol{E_R}[:,k] & = \{  \boldsymbol{y}(f_{1k};~n_1)[i] \cdot \boldsymbol{z}(f_{2k};~n_2)[j] \; | \; 0\leq i\leq q_1-1 , 0\leq j\leq q_2-1   \},
\end{align*}
and \(\boldsymbol{D}= \text{diag}(b_1,\cdots,b_r)\). Consequently, \(\mathcal{H}\boldsymbol{x}\) remains a rank \(r\) matrix for two-dimensional arrays.

The LPPG algorithm is readily adaptable for high-dimensional SSS recovery problems, with efficient implementations for Hankel matrix-vector multiplications and the application of \(\mathcal{H}^*\).

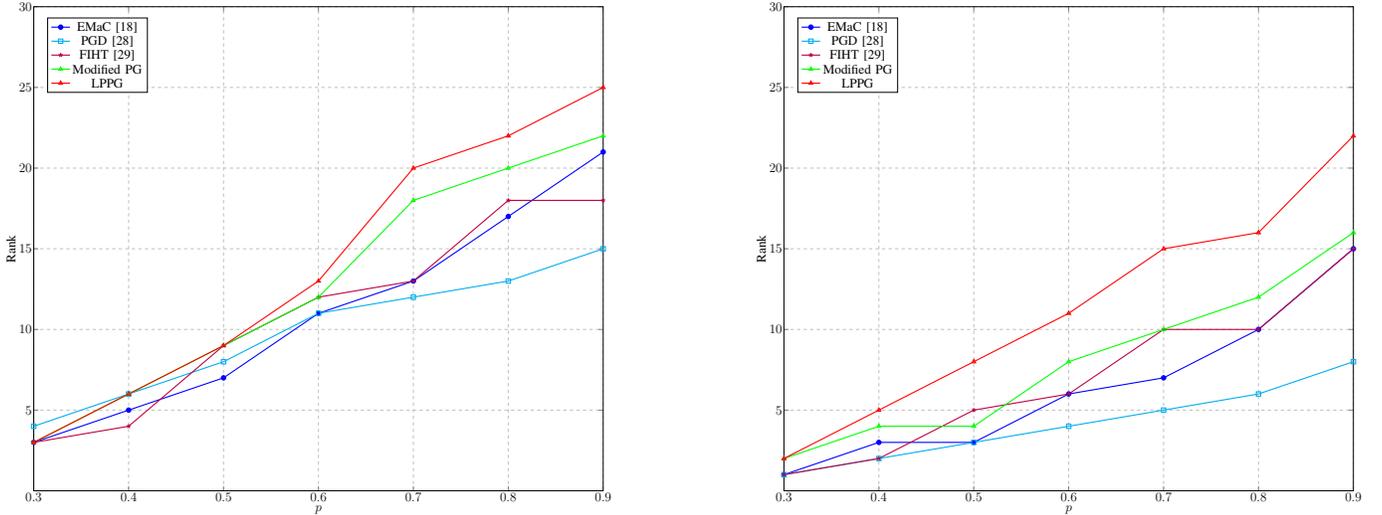
\begin{figure*}[htbp]
    \centering
    \subfigure{\resizebox{0.45\linewidth}{!}{{\scalefont{0.45}
\begin{tikzpicture}

\begin{axis}[
     width=1.1\linewidth,
    xlabel={$p$ },
    xlabel shift = 1 pt,
    ylabel={\text{Rank}},
    xmin=0.3, xmax=0.9,
    ymin=0, ymax=30,
    xtick=
    {0.3,0.4,0.5,0.6,0.7,0.8,0.9},
    ytick={5,10,15,20,25,30},
    legend style={at={(0.20,0.90)},anchor=east},
    ymajorgrids=true,
    xmajorgrids=true,
    grid style=dashed,
    font=\large
]

\addplot[
    color=blue,
    mark=*,
    line width=1.0pt,
    ]
  coordinates{(0.3,3)(0.4,5)(0.5,7)(0.6,11)(0.7,13)(0.8,17)(0.9,21)
    };
    \addlegendentry{EMaC\cite{chen2013spectral} }

\addplot[
    color=cyan,
    mark=square,
    line width=1.0pt,
    ]
     coordinates{(0.3,4)(0.4,6)(0.5,8)(0.6,11)(0.7,12)(0.8,13)(0.9,15)
    };
\addlegendentry{PGD\cite{cai2018spectral}}

\addplot[
    color=purple,
    mark=star,
    line width=1.0pt,
    ]
     coordinates {(0.3,3)(0.4,4)(0.5,9)(0.6,12)(0.7,13)(0.8,18)(0.9,18)
    };
\addlegendentry{FIHT\cite{cai2019fast}}

\addplot[
    color=green,
    mark=triangle,
    line width=1.0pt,
    ]
  coordinates {(0.3,3)(0.4,6)(0.5,9)(0.6,12)(0.7,18)(0.8,20)(0.9,22)
    };
    \addlegendentry{Modified PG}

\addplot[
    color=red,
    mark=triangle,
    line width=1.0pt,
    ]
   coordinates{(0.3,3)(0.4,6)(0.5,9)(0.6,13)(0.7,20)(0.8,22)(0.9,25)
    };
    \addlegendentry{LPPG }

\end{axis}
\end{tikzpicture}
}}} \hfill
    \subfigure{\resizebox{0.45\linewidth}{!}{{\scalefont{0.45}
\begin{tikzpicture}

\begin{axis}[
     width=1.1\linewidth,
    xlabel={$p$ },
    xlabel shift = 1 pt,
    ylabel={\text{Rank}},
    xmin=0.3, xmax=0.9,
    ymin=0, ymax=30,
    xtick=
    {0.3,0.4,0.5,0.6,0.7,0.8,0.9},
    ytick={5,10,15,20,25,30},
    legend style={at={(0.20,0.90)},anchor=east},
    ymajorgrids=true,
    xmajorgrids=true,
    grid style=dashed,
    font=\large
]

\addplot[
    color=blue,
    mark=*,
    line width=1.0pt,
    ]
  coordinates{(0.3,1)(0.4,3)(0.5,3)(0.6,6)(0.7,7)(0.8,10)(0.9,15)
    };
    \addlegendentry{EMaC\cite{chen2013spectral} }

\addplot[
    color=cyan,
    mark=square,
    line width=1.0pt,
    ]
     coordinates{(0.3,1)(0.4,2)(0.5,3)(0.6,4)(0.7,5)(0.8,6)(0.9,8)
    };
\addlegendentry{PGD\cite{cai2018spectral}}

\addplot[
    color=purple,
    mark=star,
    line width=1.0pt,
    ]
     coordinates {(0.3,1)(0.4,2)(0.5,5)(0.6,6)(0.7,10)(0.8,10)(0.9,15)
    };
\addlegendentry{FIHT\cite{cai2019fast}}

\addplot[
    color=green,
    mark=triangle,
    line width=1.0pt,
    ]
  coordinates {(0.3,2)(0.4,4)(0.5,4)(0.6,8)(0.7,10)(0.8,12)(0.9,16)
    };
    \addlegendentry{Modified PG}

\addplot[
    color=red,
    mark=triangle,
    line width=1.0pt,
    ]
   coordinates{(0.3,2)(0.4,5)(0.5,8)(0.6,11)(0.7,15)(0.8,16)(0.9,22)
    };
    \addlegendentry{LPPG }

\end{axis}
\end{tikzpicture}
}}}
    \caption{$50\%$ phase transition curves: $n_1=63, (p_1,q_1)=(32,32)$. Left: no damping in the test signals. Right: signals are generated with damping. }
    \label{fig: phase transition}
\end{figure*}

\section{Simulations} \label{sect: numerical}

This section presents numerical simulations to showcase the advantages of our LPPG method, with a specific focus on its recovery ability. The simulations are structured as follows:

\begin{enumerate}
    \item \textbf{Phase Transition Analysis:} In Section \ref*{Subsect: phase transition}, we evaluate the recovery ability of LPPG through a phase transition framework.
    
    \item \textbf{Convergence Rate without Noise:} Section \ref{Subsect: convergence rate} illustrates the fast convergence rate of LPPG facilitated by a large step-size in noise-free scenarios.
    
    \item \textbf{Performance under Heavy Noise:} In Section \ref{subsec: noise}, we address the challenging case with heavy additive noise and demonstrate the flexibility of the hyperparameter \(\beta\), which can be chosen within a wide range.
    
    \item \textbf{Weighted Norm Issue:} Section \ref{subsect: NS formulation} shows how our formulation addresses the issue of weighted norm caused by repetitive elements in matrix form.
\end{enumerate}

In all tests, the ground-truth SSS and their partial observations are generated as follows: The frequency \(f_k\) of SSS is uniformly distributed on \([0,1)\), and the phases of complex coefficients \(b_k\) are uniformly sampled on \([0,2\pi)\). The amplitudes are set to \(1 + 10^{0.5c_k}\), where \(c_k\) is uniformly distributed on \([0,1]\); cf. \cite{cai2018spectral}. For a given sample size \(m = S_pn\), \(\Omega\) is assumed to be uniformly sampled. We assess computational efficiency (number of iterations) and normalized mean squared error (NMSE, defined as \(\|\boldsymbol{x}^{\dag} - \boldsymbol{x}\|_2/\|\boldsymbol x\|_2\), where \(\boldsymbol{x}^{\dag}\) denotes the estimate returned by LPPG) to highlight the advantages of our LPPG method. We consider level-1, level-2, and level-3 cases in the data settings, with two different settings: (a) without damping factor, and (b) damping factors generated such that \(1/\tau_{1,k}\) is uniformly sampled from \([8,16]\) for \(1\leq k \leq r\) (similarly, \(1/\tau_{2,k}\) from \([16,32]\) and \(1/\tau_{3,k}\) from \([64,128]\))\cite{cai2018spectral}.

Our LPPG method is compared with benchmark algorithms, including EMaC\cite{chen2013spectral}, FIHT \cite{cai2019fast}, PGD \cite{cai2018spectral}, standard PG \cite{parikh2014proximal}, and the modified PG algorithm (MPG), which is LPPG without subspace projection. For large-size cases, we apply EMaC\cite{chen2013spectral} using the alternating projection-based implementation as suggested in \cite{chen2013spectral}. The value of \(\alpha\) is set to \(1e-20\) by default.

\subsection{Empirical Phase Transition} \label{Subsect: phase transition}

In our analysis, an algorithm is deemed to have successfully recovered a signal if the Normalized Mean Squared Error (NMSE) is less than \(10^{-3}\). We illustrate the empirical phase transition in Figure \ref{fig: phase transition}. The displayed results are averaged over \(50\) trials, and the phase transition curves represent the \(50\%\) success rate as a function of the sampling ratio \(S_p\) and rank \(r\). These curves are plotted for both undamped and damped cases. Due to computational constraints, we select \(n=63\) to effectively demonstrate the advantages of our approach. The stopping criteria are set to a maximum of \(1000\) iterations (denoted as \(K\)) and a numerical-error threshold \(\epsilon\) of \(10^{-6}\). Notably, the standard PG algorithm is excluded from the phase transition curves due to its inability to converge within 1000 iterations.

Comparatively, the damped case presents more challenges than the undamped case, necessitating a higher sampling ratio for successful recovery. Our results indicate that both the LPPG and MPG methods surpass other algorithms in terms of performance in both scenarios. Specifically, the LPPG method, especially in damped cases, showcases the necessity of the subspace projection technique when the model order is moderate relative to the signal size. This technique is critical for the success of the LPPG algorithm, as it enables further reduction of the objective function in each iteration. Moreover, the LPPG method can achieve successful recovery with a smaller sampling ratio for a fixed \(r\), an important consideration in practical applications where the sensing process is costly. These findings underscore the superiority of the LPPG method in both undamped and damped scenarios.

\begin{remark}
Given the superior recovery ability of LPPG, it is feasible to select scenarios where only LPPG is effective in subsequent tests. However, to ensure fairness, we avoid such biased comparisons and choose scenarios where all algorithms can operate, particularly in Sections \ref{Subsect: convergence rate} and \ref{subsec: noise}.
\end{remark}

\begin{figure*}[htbp]
    \centering
    \subfigure{\resizebox{0.45\linewidth}{!}{{\scalefont{0.45}
\begin{tikzpicture}

\begin{semilogyaxis}[
     width=1.1\linewidth,
    xlabel={\(\#\)Iterations },
    xlabel shift = 1 pt,
    ylabel={\text{NMSE}},
    xmin=0, xmax=100,
    ymin=1e-6, ymax=1,
    xtick=
    {0,10,20,30,40,50,60,70,80,90,100},
    ytick={1e-6,1e-5,1e-4,1e-3,1e-2,1e-1,1.0},
    legend style={at={(0.03,0.03)}, anchor=south west}, 
    ymajorgrids=true,
    xmajorgrids=true,
    grid style=dashed,
    font=\large
]

\addplot[
    color=blue,
    mark=*,
    line width=1.0pt,
    ]
  coordinates{(0,0.812)(10,0.149)(20,0.0231)(30,0.0071)(40,0.0016)(50,0.000413)(60,0.0001)(70,2.66e-5)(80,6.88e-6)(90,1.79e-6)(100,4.73e-7)
    };
    \addlegendentry{EMaC\cite{chen2013spectral} }

\addplot[
    color=cyan,
    mark=square,
    line width=1.0pt,
    ]
    coordinates{(0,0.812)(10,0.304)(20,0.183)(30,0.0991)(40,0.037)(50,0.00323)(60,0.00162)(70,0.00126)(80,0.00020)(90,5.03e-5)(100,1.78e-5)
    };   
\addlegendentry{PGD\cite{cai2018spectral}}

\addplot[
    color=purple,
    mark=star,
    line width=1.0pt,
    ]
    coordinates{(0,0.812)(10,0.035)(20,0.011)(30,0.00066)(40,0.00021)(50,7.07e-5)(60,2.47e-5)(70,8.7e-6)(80,3.0e-6)(90,1.09e-6)(100,3.0e-7)
    }; 
\addlegendentry{FIHT\cite{cai2019fast}}

\addplot[
    color=red,
    mark=triangle,
    line width=1.0pt,
    ]
    coordinates{(0,0.812)(10,0.654)(20,0.654)(30,0.6300)(40,0.6300)(50,0.63)(60,0.63)(70,0.63)(80,0.63)(90,0.63)(100,0.63)
    }; 
    \addlegendentry{standard PG\cite{parikh2014proximal} }

\addplot[
    color=green,
    mark=triangle,
    line width=1.0pt,
    ]
    coordinates{(0,0.812)(10,0.122)(20,0.0191)(30,0.0061)(40,0.0014)(50,0.000360)(60,9.09e-5)(70,2.32e-5)(80,6.03e-6)(90,1.57e-6)(100,4.0e-7)
    }; 
    \addlegendentry{Modified PG}

\addplot[
    color=red,
    mark=triangle,
    line width=1.0pt,
    ]
    coordinates{(0,0.812)(10,0.036)(20,0.0050)(30,0.0006)(40,9.73e-5)(50,1.62e-5)(60,2.76e-6)(70,4.74e-7)
    }; 
    \addlegendentry{LPPG }

\end{semilogyaxis}
\end{tikzpicture}
}}\label{fig: convergencerate_undamped}} \hfill
    \subfigure{\resizebox{0.45\linewidth}{!}{{\scalefont{0.45}
\begin{tikzpicture}

\begin{semilogyaxis}[
     width=1.1\linewidth,
    xlabel={\(\#\)Iterations },
    xlabel shift = 1 pt,
    ylabel={\text{NMSE}},
    xmin=0, xmax=100,
    ymin=1e-6, ymax=1,
    xtick=
    {0,10,20,30,40,50,60,70,80,90,100},
    ytick={1e-6,1e-5,1e-4,1e-3,1e-2,1e-1,1.0},
    legend style={at={(0.03,0.03)}, anchor=south west}, 
    ymajorgrids=true,
    xmajorgrids=true,
    grid style=dashed,
    font=\large
]

\addplot[
    color=blue,
    mark=*,
    line width=1.0pt,
    ]
    coordinates{(0,0.751)(10,0.148)(20,0.0219)(30,0.00395)(40,0.000756)(50,0.000149)(60,2.96e-5)(70,3.52e-6)(80,1.22e-6)(90,2.63e-7)
    }; 
    \addlegendentry{EMaC\cite{chen2013spectral} }

\addplot[
    color=cyan,
    mark=square,
    line width=1.0pt,
    ]
    coordinates{(0,0.751)(10,0.445)(20,0.371)(30,0.331)(40,0.321)(50,0.300)(60,0.284)(70,0.264)(80,0.246)(90,0.255)(100,0.245)
    };  
\addlegendentry{PGD\cite{cai2018spectral}}

\addplot[
    color=purple,
    mark=star,
    line width=1.0pt,
    ]
    coordinates{(0,0.751)(10,0.018)(20,0.00169)(30,0.000260)(40,6.28e-5)(50,2.23e-5)(60,8.85e-6)(70,3.52e-6)(80,1.38e-6)(90,5.5e-7)
    }; 
\addlegendentry{FIHT\cite{cai2019fast}}

\addplot[
    color=red,
    mark=triangle,
    line width=1.0pt,
    ]
    coordinates{(0,0.751)(10,0.544)(20,0.544)(30,0.544)(40,0.544)(50,0.544)(60,0.544)(70,0.544)(80,0.544)(90,0.544)(100,0.544)
    }; 
    \addlegendentry{standard PG\cite{parikh2014proximal} }

\addplot[
    color=green,
    mark=triangle,
    line width=1.0pt,
    ]
    coordinates{(0,0.751)(10,0.118)(20,0.018)(30,0.00347)(40,0.000674)(50,0.000134)(60,2.70e-5)(70,5.52e-6)(80,1.12e-6)(90,2.33e-7)
    }; 
    \addlegendentry{Modified PG}

\addplot[
    color=red,
    mark=triangle,
    line width=1.0pt,
    ]
    coordinates{(0,0.751)(10,0.020)(20,0.0021)(30,0.000246)(40,2.87e-5)(50,3.38e-6)(60,3.98e-7)
    }; 
    \addlegendentry{LPPG }

\end{semilogyaxis}
\end{tikzpicture}
}}\label{fig: convergencerate_damped}}
    \caption{Convergence rate comparison: $(n_1,n_2)=(31,31), (p_1p_2,q_1q_2)=(256,256),r=15$. Left: no damping in the test signals and $S_p=0.3$. Right: signals are generated with damping and $S_p=0.4$. }
    \label{convergencerate}
\end{figure*}

 \subsection{Fast Convergence Rate Under Noiseless Case}\label{Subsect: convergence rate}

In this experiment, we assess the convergence rate of our LPPG algorithm in an ideal, noiseless environment. The convergence rate is influenced by the parameter \(\beta\), as indicated by the inequality:
\[
\min_{i=0,\cdots,K} ~  \|\partial F(\boldsymbol{H}_{i+1})  \|_F^2 \leq \frac{(\beta-L_f)\beta^4}{(K+1)(\alpha+\beta)^2} ( F(\boldsymbol{H}_{0}) - F^*),
\]
implying that smaller values of \(\beta\) lead to a faster convergence rate. In line with \cite{davenport2016overview}, it is reasonable to set the structural constraint parameter close to zero when the observed data is accurate. Therefore, for the modified PG-based algorithms (MPG and LPPG), we set the step size \(\gamma=1/\beta=1e6\), ensuring convergence irrespective of the value of \(\beta\) since \(L_f < \beta\). Conversely, gradient-based methods like FIHT and PGD can only empirically adjust their step-size to the sampling ratio, as their step size is not directly related to the noise level, requiring a time-consuming tuning process. The alternating projection-based EMaC algorithm cannot increase the step size.

We compare the convergence rates of different algorithms in both undamped and damped settings in Figures \ref{fig: convergencerate_undamped} and \ref{fig: convergencerate_damped}, respectively, by plotting NMSE as a function of the number of iterations. The test involves a two-dimensional spectrally sparse signal \(\boldsymbol{s}\) of size \(\boldsymbol{n}=31 \times 31\), with \(\mathcal{H}\boldsymbol{s} \in \mathbb{C}^{256 \times 256}\) being a level-2 block Hankel matrix. The sampling ratios are set to \(0.3\) and \(0.4\) for undamped and damped settings, respectively, and the rank is fixed at \(15\). This setting ensures that all algorithms are operational while allowing for differentiation in performance. The results are averaged over \(50\) trials, with stopping criteria set at \(K=1000\) iterations and \(\epsilon =10^{-6}\).

In this scenario, with moderate rank and sampling ratio, the LPPG algorithm demonstrates the fastest convergence in both undamped and damped cases. The MPG also shows a quicker convergence than PGD and EMaC, highlighting the benefit of the modified PG step. However, the standard PG algorithm converges slowly due to the small step size dictated by \(L_f\), underscoring the challenge of applying standard PG to spectral compressed sensing with equally weighted norms. Our results confirm that both the subspace projection and the modified PG step are pivotal for achieving faster convergence rates. Although FIHT's convergence speed is competitive, the final result values may become unstable as it does not guarantee a reduction in its objective function in each iteration.

\subsection{Robustness to Additive Noise}\label{subsec: noise}

In practical applications, obtaining noise-free measurements is often unfeasible, making the robustness of SSS recovery algorithms to additive noise a critical factor. To assess the performance of our LPPG approach under such conditions, we introduce noise into the measurements using the following noise vector:
\[
\boldsymbol{e} = \eta \cdot \| \mathcal{P}_{\Omega}\boldsymbol{s} \| \cdot \frac{\boldsymbol{w}}{\|\boldsymbol{w}\|},
\]
where \(\boldsymbol{w}\) is a standard complex Gaussian random vector, and \(\eta\) represents the noise level. We examine a challenging scenario with \(\eta=1\), which corresponds to a Signal-to-Noise Ratio (SNR) of 0 dB.

In this experiment, we consider a three-dimensional spectrally sparse signal \(\boldsymbol{s}\) of size \(\boldsymbol{n}=15 \times 15 \times 15\), with \(\mathcal{H}\boldsymbol{s} \in \mathbb{C}^{512 \times 512}\) forming a level-3 block Hankel matrix. We set the sampling ratio to \(1.0\) in both the undamped and damped settings to focus on the denoising capability. The rank is fixed at \(10\). As in the previous section, the results are averaged over \(50\) trials. This setting is chosen to provide a fair comparison and demonstrate the utility of a large Hankel enforcement parameter \(\beta\). 

All tested algorithms terminate when either of the following conditions is met: the relative change in the Hankel matrix \(\|\boldsymbol{H}_{k+1}-\boldsymbol{H}_{k}\|_F / \|\boldsymbol{H}_{k}\|_F\) falls below \(10^{-3}\), or the maximum iteration count (1000) is reached. This stopping criterion is chosen to be relatively larger due to the high noise level and the diminishing returns in NMSE improvement with an increased number of iterations. The results are presented in Table \ref{table: noise robust}.

First of all, the EMaC algorithm, being a constrained optimization implemented via the alternating projection method, cannot vary the Hankel enforcement parameter \(\beta\), serving as a baseline for recovery accuracy. In contrast, to enhance reconstruction accuracy under noisy conditions, the LPPG method and other algorithms allow for a larger \(\beta\) value compared to noise-free scenarios. For this experiment, we select \(\beta\) values from the set \(\{1*\beta^*, 10\beta^*, 100\beta^*\}\)\footnote{Here the $\beta^*$ is the default value for each algorithm. Moreover $\beta^*=S_p$ for Enhanced matrix form (PGD and FIHT), $\beta^*=S_p*n_1n_2n_3/(p_1p_2p_3*q_1q_2q_3)$ for vector form (PG-based algorithms)  as a fair comparison.}, with the understanding that a larger \(\beta\) implies a smaller step size for the standard PG algorithm, potentially leading to its early termination. This setup not only tests the algorithms' robustness to noise but also demonstrates the impact of the parameter \(\beta\) on their performance.

\begin{table}[ht]
\caption{
    NMSE, number of iterations and time comparisons under heavy noise $SNR=0dB, (n_1,n_2,n_3)=(15,15,15), (p_1p_2p_3,q_1q_2q_3)=(512,512),r=10, S_p=1.0$.
} 
\label{table: noise robust}
\resizebox{\linewidth}{!}{%
\begin{tabular}{c|cc|cc|cc}
\hline
$\beta$                             & \multicolumn{2}{c}{$\beta=1*\beta^*$  }       &   \multicolumn{2}{c}{$\beta=10\beta^*$  }  & \multicolumn{2}{c}{$\beta=100\beta^*$  }                        \\ 
\hline
Criterion         & NMSE         & Iter             & NMSE          & Iter         & NMSE            & Iter                          \\ 
\hline
    &  \multicolumn{6}{c}{ undamped case}\\
\hline 
{\color{blue}EMaC\cite{chen2013spectral}  } & {\color{blue}0.162}      & {\color{blue}12}     & {\color{blue}0.162}        & {\color{blue}12}       & {\color{blue}0.162}      & {\color{blue}12}              \\
PGD\cite{cai2018spectral}   & 0.189      & 10    & 0.190        & 21    & 0.190    & 22                         \\
FIHT\cite{cai2019fast}   & 0.178       & 14    & 0.160          & 9       & 0.162    & 12        \\
standard PG\cite{parikh2014proximal}   & 0.650       & 1     & 0.306       & 1      & 0.185        & 1              \\
Modified PG  & 0.646      & 7   &  0.282        & 18   & 0.140     & 37         \\
{\color{red} LPPG}     & {\color{red}  0.646  }  & {\color{red}    6} & {\color{red}0.282}  & {\color{red}  12  } & {\color{red}  0.136  } & {\color{red}31}  \\
\hline
  & \multicolumn{6}{c}{ damped case}\\
\hline 
{\color{blue}EMaC\cite{chen2013spectral}  }  & {\color{blue}0.171}       & {\color{blue}14}       & {\color{blue}0.171}         & {\color{blue}14}       & {\color{blue}0.171}         &{\color{blue}14}                  \\
PGD\cite{cai2018spectral}   & 0.182      & 10   & 0.167        & 40    & 0.168    & 43                         \\
FIHT\cite{cai2019fast}   & 0.182      & 14    & 0.169          & 14       & 0.171    & 12       \\
standard PG\cite{parikh2014proximal}   & 0.651      & 1    & 0.308        & 1     & 0.189        & 1              \\
Modified PG  &0.646       & 7   &  0.283        & 23     & 0.143      & 43       \\
{\color{red} LPPG}     & {\color{red}  0.646  }  & {\color{red}    7} & {\color{red}0.283}  & {\color{red}  20 } & {\color{red}  0.141  } & {\color{red}36}  \\
\hline
\end{tabular}
}
\end{table}

The findings presented in Table \ref{table: noise robust} allow us to draw several key observations regarding the recovery accuracy under heavy additive noise:

\begin{enumerate}
    \item \textit{Improved Recovery with Larger \(\beta\):} It becomes clear that using larger Hankel enforcement parameters (\(\beta\)) generally leads to improved recovery accuracy. This trend is consistent across different algorithms, underscoring the significance of adjusting \(\beta\) in noisy environments.

    \item \textit{Superiority of MPG and LPPG at Higher \(\beta\):} Particularly when \(\beta = 100\), both MPG and LPPG demonstrate superior performance compared to other algorithms, with LPPG achieving the highest level of accuracy. Remarkably, even the standard PG with a single iteration competes effectively with other algorithms under these conditions. Conversely, FIHT and PGD only surpass the performance of EMaC when \(\beta\) is increased, highlighting the critical role of the Hankel constraint in algorithmic performance.

    \item \textit{Effectiveness of Stopping Criteria:} As discussed in Section \ref{subsect: stopping criteria}, the term \(\beta^2 \mathcal{H}\mathcal{L}\mathcal{H}^*(\boldsymbol{H}_{k,\frac{1}{2}}- \boldsymbol{H}_{k+1})  \in \partial F(\boldsymbol{H}_{k+1})\) is relevant for our PG-based algorithms. The stopping criteria, therefore, serve as a reliable standard for halting iterative algorithms. This aspect is particularly crucial for algorithms like PGD and FIHT, which may otherwise oscillate around the optimal solution.
\end{enumerate}

In summary, these results collectively indicate that our LPPG method stands out as a robust and efficient algorithm for recovering SSS in scenarios characterized by heavy additive noise. This robustness, coupled with the method's efficiency, underscores its potential utility in practical applications where noise is an unavoidable factor.

\subsection{The Benefit of Formulation}\label{subsect: NS formulation}

A key contribution of our work is making Proximal Gradient (PG) algorithms practical for vector formulations through the modified PG. This section showcases the advantages of our formulations in vector form, which align with the original problem setting and enhance reconstruction accuracy as signal size increases.

Previous studies \cite{gillard2010cadzow,wang2021fast} have observed that Cadzow's basic algorithm may lead to biased reconstructions due to repetitive elements in the enhanced matrix. To demonstrate this, we present the Normalized Mean Squared Error (NMSE) plots for the FIHT, PGD, EMaC, and LPPG algorithms (MPG and standard PG share the same data fidelity item as LPPG) in 1D SSS recovery for signal sizes \(n \in \{17, 33, 65, 129 \}\). As signal size increases, the impact of repetitive elements becomes more pronounced. The results, averaged over 50 trials, consider fully observed signals with noise and a rank of 2. The stopping criteria are set at \(K=1000\) and \(\epsilon =10^{-6}\).

In addition to presenting NMSE as a direct metric for recovery ability across different signal sizes, we also include a bar chart depicting the Relative Mean Squared Error (RMSE) related to LPPG (defined as \( \|\text{NMSE}_{\text{LPPG}} -\text{NMSE}\|_2/\| \text{NMSE}_{\text{LPPG}} \|_2\), where \(\text{NMSE}_{\text{LPPG}}\) denotes the NMSE returned by LPPG).

\begin{figure*}[htbp]
    \label{fig: formulation}
    \centering
    \subfigure{\resizebox{0.45\linewidth}{!}{\begin{tikzpicture}
\begin{axis}[
    ybar,
    ylabel={NMSE},
    ymin=1e-2, ymax=1,
    symbolic x coords={17, 33, 65, 129},
    xtick=data,
    ymin=0,
   legend style={at={(0.03,0.99)}, anchor=north west, font=\tiny},
    axis y line*=left,
]

\addplot [sharp plot, color=blue,
    mark=*,dashed,
    mark size=1pt] coordinates {(17,0.684) (33,0.412) (65,0.306) (129,0.202)};
\addlegendentry{EMaC\cite{chen2013spectral}}

\addplot [sharp plot, color=cyan,
   dashed,  mark=square, mark size=1pt, ] coordinates {(17,0.675) (33,0.400) (65,0.297) (129,0.201)};
\addlegendentry{PGD\cite{cai2018spectral}}

\addplot [sharp plot, color=purple,
    mark=star,dashed,
    mark size=1pt] coordinates {(17,0.674) (33,0.400) (65,0.297) (129,0.201)};
\addlegendentry{FIHT\cite{cai2019fast}}

\addplot [sharp plot, color=red,
    mark=triangle,mark size=1pt, thick] coordinates {(17,0.636) (33,0.358) (65,0.244) (129,0.168)};
\addlegendentry{LPPG}

\end{axis}

\begin{axis}[
    ybar,
    ylabel={RMSE},
    ymin=1e-2, ymax=1,
    symbolic x coords={17, 33, 65, 129},
    xtick=data,
    ymin=0, ymax=1.0,
    legend style={at={(0.98,0.98)},anchor=north east, font=\tiny},
    axis y line*=right,
]
\addplot[fill=none,draw=blue] coordinates {(17,0.068) (33,0.218) (65,0.239) (129,0.194)};
\addlegendentry{LPPG-EMaC}

\addplot[fill=none,draw=cyan] coordinates {(17,0.063) (33,0.193) (65,0.215) (129,0.192)};
\addlegendentry{LPPG-PGD}

\addplot[fill=none,draw=purple] coordinates {(17,0.063) (33,0.195) (65,0.214) (129,0.192)};
\addlegendentry{LPPG-FIHT}

\end{axis}

\end{tikzpicture}}\label{fig: formulation undamped}} \hfill
    \subfigure{\resizebox{0.45\linewidth}{!}{\begin{tikzpicture}
\begin{axis}[
    ybar,
    ylabel={NMSE},
    ymin=1e-2, ymax=1,
    symbolic x coords={17, 33, 65, 129},
    xtick=data,
    ymin=0,
   legend style={at={(0.03,0.99)}, anchor=north west, font=\tiny},
    axis y line*=left,
]


\addplot [sharp plot, color=blue,
    mark=*,dashed,
    mark size=1pt] coordinates {(17,0.659) (33,0.493) (65,0.427) (129,0.285)};
\addlegendentry{EMaC\cite{chen2013spectral}}

\addplot [sharp plot, color=cyan,
   dashed,  mark=square, mark size=1pt, ] coordinates {(17,0.664) (33,0.514) (65,0.442) (129,0.294)};
\addlegendentry{PGD\cite{cai2018spectral}}
\addplot [sharp plot, color=purple,
    mark=star,dashed,
    mark size=1pt] coordinates {(17,0.664) (33,0.515) (65,0.445) (129,0.293)};
\addlegendentry{FIHT\cite{cai2019fast}}

\addplot [sharp plot, color=red,
    mark=triangle,mark size=1pt, thick] coordinates {(17,0.590) (33,0.406) (65,0.305) (129,0.179)};
\addlegendentry{LPPG}

\end{axis}

\begin{axis}[
    ybar,
    ylabel={RMSE},
    ymin=1e-2, ymax=1,
    symbolic x coords={17, 33, 65, 129},
    xtick=data,
    ymin=0, ymax=1.0,
    legend style={at={(0.98,0.98)},anchor=north east, font=\tiny},
    axis y line*=right,
]
\addplot[fill=none,draw=blue] coordinates {(17,0.118) (33,0.228) (65,0.397) (129,0.625)};
\addlegendentry{LPPG-EMaC}

\addplot[fill=none,draw=cyan] coordinates {(17,0.125) (33,0.266) (65,0.446) (129,0.634)};
\addlegendentry{LPPG-PGD}
\addplot[fill=none,draw=purple] coordinates {(17,0.125) (33,0.267) (65,0.459) (129,0.634)};
\addlegendentry{LPPG-FIHT}

\end{axis}

\end{tikzpicture}}\label{fig: formulation damped}}
    \caption{Convergence rate comparison: $n_1={17,33,65,129}, (p_1,q_1)={(9,9),(17,17),(33,33),(65,65)},r=2,S_p=1.0,$SNR$ =0$dB. Left: no damping in the test signals. Right: signals are generated with damping. }
    \label{figure: formulation}
\end{figure*}
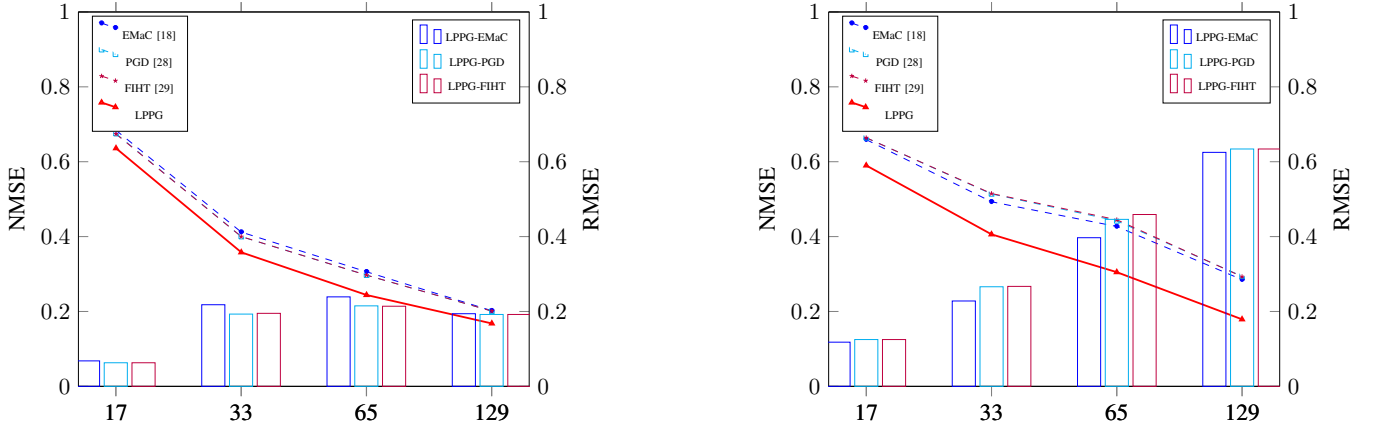

Figure \ref{figure: formulation} reveals two key observations. Firstly, all algorithms show improved reconstruction accuracy as signal size increases, likely due to the signal's increased sparsity in the frequency domain. Secondly, the LPPG algorithm consistently outperforms other algorithms in all cases, with RMSE relative to LPPG growing as signal size increases. This trend is attributable to the more evident repetitive elements in the enhanced matrix at larger signal sizes. However, the LPPG algorithm remains unaffected by this issue due to its reliance on the Hankel structure. This advantage is particularly noticeable in the damped case.

\section{Conclusions} \label{sect: conclusion}
In this paper, we introduced the Low-rank Projected Proximal Gradient (LPPG) method, a novel approach for the recovery of spectrally sparse signals. This method is based on a non-convex optimization formulation involving a low-rank Hankel matrix.

The core innovation of the LPPG method lies in its iterative process, which comprises two distinct steps in each iteration. The first step involves optimization with respect to all variables, excluding the column and row sub-spaces of the Hankel matrix. The second step simultaneously updates these sub-spaces along with all other variables. Each step is meticulously tailored to leverage the inherent structures of the problem, optimizing the recovery process.

A significant advantage of the LPPG method is its guaranteed convergence, a feature that not only enhances its theoretical appeal but also bolsters its practical utility. Through numerical simulations, we have demonstrated that LPPG exhibits a faster convergence rate and more accurate reconstruction compared to several established benchmark algorithms. This attribute renders it particularly advantageous for handling large-scale problems with heavy noise, a common challenge in practical applications.

Furthermore, our modifications to the Proximal Gradient (PG) method have opened new avenues for its application in vector-form compressed sensing. Looking ahead, we are excited about the prospect of integrating other accelerated PG techniques into this framework, potentially unlocking even more efficient and powerful solutions for spectrally sparse signal recovery.

\section{Proofs} \label{sect: proofs}
In this section, we provide mathematical proofs for the theoretical results stated in Proposition \ref{prop:Lipschitz-Constant}, Proposition \ref{pro: gradient_h}, Proposition \ref{Theo_SP_reduction}, and Theorem \ref{theo: convergence}. 

We begin by proving the following lemma, which offers a useful formula for computing the gradient of a function minimized with respect to its arguments.

\begin{lemma} \label{lemma: innergradient}
Let \( h:\mathbb{C}^{p\times q}\times\mathbb{C}^n\rightarrow\mathbb{R} \) be a continuously differentiable function, and let \( f:\mathbb{C}^{p \times q}\rightarrow\mathbb{R} \) be defined as \( f(\boldsymbol{H})=\min_{\boldsymbol{x}\in\mathbb{C}^n} h(\boldsymbol{H},\boldsymbol{x}) \). Assume that for each \( \boldsymbol{H}\in\mathbb{C}^{p \times q} \), there exists a minimizer \( \boldsymbol{x}^{\star}(\boldsymbol{H}) \).

Then, for any \( \boldsymbol{H}\in\mathbb{C}^{p\times q} \), the gradient of \( f \) with respect to \( \boldsymbol{H} \) exists and is given by:
\[
\nabla f(\boldsymbol{H})=\nabla_{\boldsymbol{H}} h(\boldsymbol{H},\boldsymbol{x}^{\star}),
\]
\end{lemma}

\begin{proof}
To prove this lemma, we utilize the chain rule of calculus. Let \( \boldsymbol{x}^{\star} \) be the minimizer of \( h(\boldsymbol{H},\boldsymbol{x}) \), such that \( f(\boldsymbol{H})=h(\boldsymbol{H},\boldsymbol{x}^{\star}) \). We then have:
\begin{align*}
\nabla f(\boldsymbol{H}) &= \nabla h(\boldsymbol{H},\boldsymbol{x}^{\star}) \\
&= \nabla_{\boldsymbol{H}} h(\boldsymbol{H},\boldsymbol{x}^{\star}) + \nabla_{\boldsymbol{x}} h(\boldsymbol{H},\boldsymbol{x}^{\star}) \nabla_{\boldsymbol{H}} \boldsymbol{x}^{\star} \\
&= \nabla_{\boldsymbol{H}} h(\boldsymbol{H},\boldsymbol{x}^{\star}),
\end{align*}
where the second equality is derived from the chain rule, and the third equality follows from the optimality condition \( \nabla_{\boldsymbol{x}} h(\boldsymbol{H},\boldsymbol{x}^{\star})=\boldsymbol{0} \).

Hence, we have proven that \( \nabla f(\boldsymbol{H})=\nabla_{\boldsymbol{H}} h(\boldsymbol{H},\boldsymbol{x}^{\star}) \).
\end{proof}

\subsection{Proof of Prop. \ref{prop:Lipschitz-Constant}} \label{subsection: Lipschitzconstant}
We adapted a modification to the non-convex non-smooth unconstrained optimization problem by transforming it into a single variable problem \( F(\boldsymbol{H}) = f(\boldsymbol{H}) + g(\boldsymbol{H}) \). A key advantage of this approach is achieving a better upper bound for \( L_{\nabla f} \) as \( \beta \), which is independent of the signal size. The proof is as follows:

\begin{proof}
Firstly, consider the gradient with respect to \( \boldsymbol{x} \):
\begin{align}
    \label{eq: gradient_x}
   & \nabla_{\boldsymbol{x}} \big(  
    \frac{1}{2} \| \boldsymbol{s} - \mathcal{P}_{\Omega}\boldsymbol{x}\|^2  
     + \frac{\beta}{2} \| \boldsymbol{H} - \mathcal{H}\boldsymbol{x}\|^{2}_{F}+\frac{\alpha}{2}\|\boldsymbol{x}\|^2 
      \big) \nonumber \\
    & =  \underbrace{(\alpha\mathcal{I}+ \mathcal{P}_{\Omega}^*\mathcal{P}_{\Omega}+\beta\mathcal{H}^*\mathcal{H})^{-1}}_{\mathcal{L}}   \boldsymbol{x} - \mathcal{P}_{\Omega}^* \boldsymbol{s}
    - \beta\mathcal{H}^* \boldsymbol{H}.
\end{align}

Setting this gradient to zero yields the closed-form solution for \( \boldsymbol{x}^{\star} \):
\[ \boldsymbol{x}^{\star} = \mathcal{L} ( \mathcal{P}_{\Omega}^*\boldsymbol{s} + \beta\mathcal{H}^*\boldsymbol{H} ). \]

Substituting \( \boldsymbol{x}^{\star} \) into \( h(\boldsymbol{H},\boldsymbol{x}) \), we get:
\[ f(\boldsymbol{H}) = h(\boldsymbol{H},\boldsymbol{x}^{\star}) = \frac{1}{2} \| \boldsymbol{s} - \mathcal{P}_{\Omega}\boldsymbol{x}^{\star} \|^2 + \frac{\beta}{2} \| \boldsymbol{H} - \mathcal{H}\boldsymbol{x}^{\star} \|^{2}_{F} + \frac{\alpha}{2}\|\boldsymbol{x}\|^2. \]

Applying Lemma \ref{lemma: innergradient}, we have:
\begin{align}
    \nabla f(\boldsymbol{H}) & =\nabla_{\boldsymbol{H}}h(\boldsymbol{H},\boldsymbol{x}^{\star})   = \beta (\boldsymbol{H}-\mathcal{H}\boldsymbol{x}^{\star})   \nonumber \\ 
    & = \beta \boldsymbol{H} - \beta \mathcal{H}\mathcal{L} ( \mathcal{P}_{\Omega}^*\boldsymbol{s} + \beta\mathcal{H}^*\boldsymbol{H} )
\end{align}

The Hessian matrix of \( f(\boldsymbol{H}) \) is then \( \beta \mathcal{I} - \beta^2 \mathcal{H} \mathcal{L} \mathcal{H}^* \). Since:
\begin{align}\label{eq:Lipschitz-constant-operator-norm}
    0 \leq \| \beta^2 \mathcal{H} \mathcal{L}\mathcal{H}^* \| = \beta \| \beta \mathcal{H} (\alpha\mathcal{I} + \mathcal{P}_{\Omega}\mathcal{P}_{\Omega}^* + \beta\mathcal{H}^*\mathcal{H})^{-1} \mathcal{H}^* \| < \beta,
\end{align}
we conclude that \( L_{\nabla f} < \beta \).\footnote{The operator norm is used in \eqref{eq:Lipschitz-constant-operator-norm}. It is the maximum eigenvalue of the corresponding matrix representation.}
\end{proof}

\subsection{Proof of Prop. \ref{pro: gradient_h}} \label{subsection: proofofclosedform}
In this proof, we aim to efficiently solve the low-rank subspace projection step by transforming \(h(\boldsymbol{\Sigma},\boldsymbol{x})\) into a single variable function \(f(\boldsymbol{\Sigma})\). We will detail the procedure to find the optimal value for the low-rank subspace projection. 

\begin{proof}
Consider the optimization problem:
\begin{align*}
    L(\boldsymbol{\Sigma}) &= \min_{\boldsymbol{x}}~ h(\boldsymbol{\Sigma},\boldsymbol{x}) \\
    &= \min_{\boldsymbol{x}}~ \frac{1}{2} \| \boldsymbol{s} - \mathcal{P}_{\Omega}\boldsymbol{x} \|^2 
    + \frac{\beta}{2} \| \boldsymbol{U} \boldsymbol{\Sigma} \boldsymbol{V}^{\mathsf{H}} - \mathcal{H}\boldsymbol{x} \|^2_F + \frac{\alpha}{2}\|\boldsymbol{x}\|^2.
\end{align*}

The optimal \( \boldsymbol{x} \), denoted as \( \boldsymbol{x}^{\dagger} \), for any given \( \boldsymbol{\Sigma} \), is found by solving the inner least squares problem \( h(\boldsymbol{\Sigma},\boldsymbol{x}) \).

By differentiating with respect to \( \boldsymbol{x} \), we obtain:
\begin{align}\label{eq: partial_x}
    \frac{\partial}{\partial \boldsymbol{x}} h(\boldsymbol{\Sigma},\boldsymbol{x})
    &= (\alpha\mathcal{I}+\mathcal{P}_{\Omega}^*\mathcal{P}_{\Omega}+\beta\mathcal{H}^*\mathcal{H}) \boldsymbol{x} - \mathcal{P}_{\Omega}^* \boldsymbol{s} - \beta\mathcal{H}^* \mathcal{P}_{\boldsymbol{U},\boldsymbol{V}} \boldsymbol{\Sigma}.
\end{align} 

Setting \eqref{eq: partial_x} to zero yields the closed-form solution for \( \boldsymbol{x^\dagger} \):
\begin{equation}
\boldsymbol{x^\dagger}=\underbrace{(\alpha\mathcal{I}+ \mathcal{P}_{\Omega}^*\mathcal{P}_{\Omega}+\beta\mathcal{H}^*\mathcal{H})^{-1}}_{\mathcal{L}}  
   (\mathcal{P}_{\Omega}^*\boldsymbol{s} + \beta\mathcal{H}^*\boldsymbol{U}\boldsymbol{\Sigma}\boldsymbol{V}^{\mathsf{H}} ). \nonumber
\end{equation}

Applying Lemma \ref{lemma: innergradient}, the gradient of \( f \) with respect to \( \boldsymbol{\Sigma} \) is:
\begin{align}
    \nabla f(\boldsymbol{\Sigma}) &= \nabla_{\boldsymbol{\Sigma}} h(\boldsymbol{\Sigma}, \boldsymbol{x}^{\dagger}) \nonumber \\
    &= \beta \mathcal{P}^*_{\boldsymbol{U},\boldsymbol{V}} \mathcal{P}_{\boldsymbol{U},\boldsymbol{V}} \boldsymbol{\Sigma} - \beta \mathcal{P}^*_{\boldsymbol{U},\boldsymbol{V}} \mathcal{H}\boldsymbol{x^\dagger} \nonumber \\
    &= \beta \boldsymbol{\Sigma} - \beta \mathcal{P}^*_{\boldsymbol{U},\boldsymbol{V}} \mathcal{H} \mathcal{L} (\mathcal{P}_{\Omega}^*\boldsymbol{s} + \beta\mathcal{H}^* \mathcal{P}_{\boldsymbol{U},\boldsymbol{V}}\boldsymbol{\Sigma} ).
\end{align}

The solution to the linear equation 
\begin{equation} \label{eq:linear eq}
    [\mathcal{I}- \beta \mathcal{P}^*_{\boldsymbol{U},\boldsymbol{V}} \mathcal{H} \mathcal{L} \mathcal{H^*}\mathcal{P}_{\boldsymbol{U},\boldsymbol{V}}]\boldsymbol{\Sigma} = \mathcal{P}^*_{\boldsymbol{U},\boldsymbol{V}} \mathcal{H} \mathcal{L}\mathcal{P}^*_{\Omega}\boldsymbol{s}.
\end{equation}
provides the Proposition \ref{pro: gradient_h} results:
\begin{align}
    \boldsymbol{x}^{\star} &= (\alpha\mathcal{I}+ \mathcal{P}_{\Omega}^*\mathcal{P}_{\Omega}+\beta\mathcal{H}^*\mathcal{H})^{-1} (\mathcal{P}_{\Omega}^*\boldsymbol{s} + \beta\mathcal{H}^* \mathcal{P}_{\boldsymbol{U},\boldsymbol{V}} \boldsymbol{\Sigma}^\star ), \nonumber \\
    \boldsymbol{\Sigma}^{\star} &= \left( \mathcal{I}- \beta \mathcal{P}^*_{\boldsymbol{U},\boldsymbol{V}} \mathcal{H} \mathcal{L} \mathcal{H^*}\mathcal{P}_{\boldsymbol{U},\boldsymbol{V}} \right)^{-1} \left( \mathcal{P}^*_{\boldsymbol{U},\boldsymbol{V}} \mathcal{H} \mathcal{L}\mathcal{P}^*_{\Omega} \boldsymbol{s} \right).
\end{align}
\end{proof}

\subsection{Proof of Prop. \ref{Theo_SP_reduction}} \label{proof_prop_2}
The proposition asserts that the low-rank subspace projection step guarantees sufficient decrease in the objective function. The formal statement is as follows:
\begin{equation}
    F(\boldsymbol{H}_{k,\frac{1}{2}}) \leq F(\boldsymbol{H}_k) - \left(\frac{\alpha}{\alpha+\beta}\right) \| \boldsymbol{H}_{k,\frac{1}{2}} - \boldsymbol{H}_{k}\|_F^2.
\end{equation}
\begin{proof}
Define the mapping \( \mathcal{M} \) as follows:
\begin{equation}
    \mathcal{M} = \left( \mathcal{I} - \beta \mathcal{P}^*_{\boldsymbol{U},\boldsymbol{V}} \mathcal{H} \mathcal{L} \mathcal{H}^*\mathcal{P}_{\boldsymbol{U},\boldsymbol{V}} \right).
\end{equation}
Then, consider the following inequality chain:
\begin{align} \label{eq: suffcient_decrease_SP}
    \| \boldsymbol{H}_{k,\frac{1}{2}} - \boldsymbol{H}_{k}  \|_F^2 &= \|\boldsymbol{\Sigma}_{k,\frac{1}{2}}-\boldsymbol{\Sigma}_{k}\|_F^2 \nonumber \\ 
    &\overset{(a)}{\leq} \left(\frac{\alpha+\beta}{\alpha}\right) \| \boldsymbol{\Sigma}_{k,\frac{1}{2}}-\boldsymbol{\Sigma}_{k} \|_{\mathcal{M}}^2 \nonumber \\
    &\overset{(b)}{\leq} \left(\frac{\alpha+\beta}{\alpha}\right) \left( F(\boldsymbol{H}_k) - F(\boldsymbol{H}_{k,\frac{1}{2}}) \right).
\end{align}
Inequality (a) follows from the fact that \( \frac{\alpha}{\alpha+\beta} \) is the smallest eigenvalue of the positive definite mapping \( \mathcal{M} \). Inequality (b) arises from the optimal solution of the least squares problem in the low-rank subspace projection step.
\end{proof}

\subsection{Proof of Theorem \ref{theo: convergence}} \label{subsection：convergence}
Our non-convex proximal gradient methods, featuring two monotonously decreasing steps, converge to critical points. This convergence is primarily attributed to the modified PG step, which demonstrates that $\|\partial F (\boldsymbol{H}_{k+1}) \|_F \leq (1/\gamma + L_{\nabla f})\| \boldsymbol{H}_{k+1} - \boldsymbol{H}_{k,\frac{1}{2}}\|_F$, with the latter term approaching zero. We also provide a convergence speed analysis. The proof can be divided into three key steps:

\begin{enumerate}
    \item \textbf{Sufficient Decrease of the PG Step:} In Algorithm \ref{LPPG}, observe that:
    \begin{align}
        \boldsymbol{H}_{k+1} &= \arg \min_{\boldsymbol{H}} \langle \nabla f (\boldsymbol{H}_{k,\frac{1}{2}}), \boldsymbol{H} - \boldsymbol{H}_{k,\frac{1}{2}} \rangle \nonumber \\
        &+ \frac{1}{2\gamma} \| \boldsymbol{H} - \boldsymbol{H}_{k,\frac{1}{2}} \|_F^2 + g(\boldsymbol{H}). \label{eq: optimal_quardratic}
    \end{align}
    Therefore, we have:
        \begin{align}
           & \langle \nabla f (\boldsymbol{H}_{k,\frac{1}{2}}) , \boldsymbol{H}_{k+1}-\boldsymbol{H}_{k,\frac{1}{2}} \rangle \nonumber  \\ 
           &+ \frac{1}{2\gamma } \|  \boldsymbol{H}_{k+1}-\boldsymbol{H}_{k,\frac{1}{2}}  \|_F^2 + g(\boldsymbol{H}_{k+1})  \nonumber  \\
           &  \leq g(\boldsymbol{H}_{k,\frac{1}{2}}).
        \end{align}
        From the Lipschitz continuous of $\nabla f$ we have 
        \begin{align}
            & F(\boldsymbol{H}_{k+1})  \leq g(\boldsymbol{H}_{k+1}) + f(\boldsymbol{H}_{k,\frac{1}{2}}) \\ \nonumber
            & + \langle \nabla f (\boldsymbol{H}_{k,\frac{1}{2}}) , \boldsymbol{H}_{k+1}-\boldsymbol{H}_{k,\frac{1}{2}} \rangle   + \frac{L_f}{2 } \|  \boldsymbol{H}_{k+1}-\boldsymbol{H}_{k,\frac{1}{2}}  \|_F^2    \\ \nonumber
           & \leq  g(\boldsymbol{H}_{k,\frac{1}{2}})  - \langle \nabla f (\boldsymbol{H}_{k,\frac{1}{2}}) , \boldsymbol{H}_{k+1}-\boldsymbol{H}_{k,\frac{1}{2}} \rangle  \\ \nonumber 
           & - \frac{1}{2\gamma } \|  \boldsymbol{H}_{k+1}-\boldsymbol{H}_{k,\frac{1}{2}}  \|_F^2  + f(\boldsymbol{H}_{k,\frac{1}{2}}) \nonumber \\ 
           &  + \langle \nabla f (\boldsymbol{H}_{k,\frac{1}{2}}) , \boldsymbol{H}_{k+1}-\boldsymbol{H}_{k,\frac{1}{2}} \rangle - \frac{L_f}{2 } \|  \boldsymbol{H}_{k+1}-\boldsymbol{H}_{k,\frac{1}{2}}  \|_F^2 \nonumber \\
            & = F(\boldsymbol{H}_{k,\frac{1}{2}}) - \frac{\frac{1}{\gamma}-L_f}{2} \|  \boldsymbol{H}_{k+1}-\boldsymbol{H}_{k,\frac{1}{2}}  \|_F^2. \label{eq: sufficient_decrease}
        \end{align}

    \item \textbf{${H}^*$ is a critical point: $0 \in \partial F(\boldsymbol{H}^*)$}
    
    Due to the decrease property of the subspace projection step, we have:
    \begin{equation}
        F(\boldsymbol{H}_{k+1,\frac{1}{2}}) \leq F(\boldsymbol{H}_{k+1}) \leq F(\boldsymbol{H}_{k,\frac{1}{2}}).
    \end{equation}
    Thus:
    \begin{equation}
        F(\boldsymbol{H}_{k+1,\frac{1}{2}}) \leq F(\boldsymbol{H}_0), \quad F(\boldsymbol{H}_{k+1}) \leq F(\boldsymbol{H}_0)
    \end{equation}
    for all $k$. It is evident that $\{\boldsymbol{H}_k\}$ and $\{\boldsymbol{H}_{k,\frac{1}{2}}\}$ are bounded. Therefore, $\{\boldsymbol{H}_k\}$ has accumulation points. As $F(\boldsymbol{H}_k)$ is decreasing, $F$ has the same value at all the accumulation points. Let this value be $F^*$. Based on \eqref{eq: sufficient_decrease}, we have:
    \begin{align}
        \frac{\frac{1}{\gamma}-L_f}{2} \| \boldsymbol{H}_{k+1}-\boldsymbol{H}_{k,\frac{1}{2}} \|_F^2 & \leq F(\boldsymbol{H}_{k,\frac{1}{2}}) - F(\boldsymbol{H}_{k+1}) \nonumber \\
        & \leq  F(\boldsymbol{H}_{k}) - F(\boldsymbol{H}_{k+1}).
    \end{align}
    Summing over $k = 0,\dots,K$, we obtain:
      \begin{align}\label{eq: sum}
           & \frac{\frac{1}{\gamma}-L_f}{2} \sum_{k=0}^K \|  \boldsymbol{H}_{k+1}-\boldsymbol{H}_{k,\frac{1}{2}}  \|_F^2  \nonumber \\
           & \leq F(\boldsymbol{H}_{0}) - F(\boldsymbol{H}_{k})\leq F(\boldsymbol{H}_{0}) - F^*\leq \infty. 
        \end{align}
    As $L_f < \frac{1}{\gamma}$, we have:
    \begin{equation}
        \lim_{k \rightarrow \infty} \| \boldsymbol{H}_{k+1}-\boldsymbol{H}_{k,\frac{1}{2}} \|_F^2 \rightarrow 0. \label{eq: HH_converge}
    \end{equation}
    From the optimality condition of \eqref{eq: optimal_quardratic}, we have:
    \begin{align} \label{eq: partialF}
        0 \in \nabla f(\boldsymbol{H}_{k,\frac{1}{2}}) + \frac{1}{\gamma} (\boldsymbol{H}_{k+1}-\boldsymbol{H}_{k,\frac{1}{2}}) + \partial g(\boldsymbol{H}_{k+1}).
    \end{align}
    Hence:
    \begin{equation} 
        -\nabla f(\boldsymbol{H}_{k,\frac{1}{2}}) + \nabla f(\boldsymbol{H}_{k+1}) - \frac{1}{\gamma} (\boldsymbol{H}_{k+1}-\boldsymbol{H}_{k,\frac{1}{2}}) \in \partial F(\boldsymbol{H}_{k+1}).
    \end{equation}
    Moreover, we get the key result: 
   \begin{empheq}[box=\fbox]{align}
      &  \| -\nabla f(\boldsymbol{H}_{k,\frac{1}{2}}) + \nabla f(\boldsymbol{H}_{k+1})- \frac{1}{\gamma} (\boldsymbol{H}_{k+1}-\boldsymbol{H}_{k,\frac{1}{2}})\| \nonumber \\
           & \leq (\frac{1}{\gamma} + L_f) \| \boldsymbol{H}_{k+1}-\boldsymbol{H}_{k,\frac{1}{2}} \| \rightarrow 0 \quad \text{as} \quad k \rightarrow \infty.
   \end{empheq}
    Let $\boldsymbol{H}^*$ be any accumulation point of $\{\boldsymbol{H}_{k,\frac{1}{2}}\}$, say $\{ \boldsymbol{H}_{k_j,\frac{1}{2}}\} \rightarrow \boldsymbol{H}^*$ as $j \rightarrow \infty$. From \eqref{eq: HH_converge}, we have $\{ \boldsymbol{H}_{k_j +1 }\} \rightarrow \boldsymbol{H}^*$ as $j \rightarrow \infty$. Since $f$ is continuously differentiable and $g$ is lower semicontinuous, we have:
    \begin{equation} \label{eq: Flimit}
        \lim_{j \rightarrow \infty} F(\boldsymbol{H}_{k_j+1}) = F(\boldsymbol{H}^*).
    \end{equation}
    Therefore, combining $\{ \boldsymbol{H}_{k_j+1}\} \rightarrow \boldsymbol{H}^*$, \eqref{eq: partialFlimit}, \eqref{eq: boundgradient}, and \eqref{eq: Flimit}, we conclude:
    \begin{equation}
        0 \in \partial F(\boldsymbol{H}^*).
    \end{equation}

        \item \textbf{Convergence speed} 
      Given that \eqref{eq:SC_F} implies $\beta^2 \mathcal{H}\mathcal{L}\mathcal{H}^*(\boldsymbol{H}_{k,\frac{1}{2}} - \boldsymbol{H}_{k+1}) \in \partial F(\boldsymbol{H}_{k+1})$, we can establish the convergence speed as follows:
        \begin{align*}
        &  \min_{i=0,\cdots,K} ~  \|\partial F(\boldsymbol{H}_{i+1})  \|_F^2 \le \frac{\beta^4}{(\alpha+\beta)^2} \cdot \min_{i=0,\cdots,K} ~ \| 
        \boldsymbol{H}_{i,\frac{1}{2}}- \boldsymbol{H}_{i+1} \|^2 _F \\
        & \le \frac{\beta^4}{(\alpha+\beta)^2} \cdot \frac{1}{K+1}\sum_{k=0}^K \|  \boldsymbol{H}_{k,\frac{1}{2}}-\boldsymbol{H}_{k+1}  \|_F^2 \\
        & \leq  \frac{c_0}{K+1} ( F(\boldsymbol{H}_{0}) - F^*) ,
        \end{align*}
        where $c_0= \frac{(\beta-L_f)\beta^4}{(\alpha+\beta)^2}$.
        
         In conclusion, we finish the proof.

\end{enumerate}

 
\bibliographystyle{IEEEtran}
\bibliography{refs}


\clearpage

\end{document}